\documentclass[sigconf,9pt]{acmart}
\makeatletter
\def\@ACM@checkaffil{
    \if@ACM@instpresent\else
    \ClassWarningNoLine{\@classname}{No institution present for an affiliation}%
    \fi
    \if@ACM@citypresent\else
    \ClassWarningNoLine{\@classname}{No city present for an affiliation}%
    \fi
    \if@ACM@countrypresent\else
        \ClassWarningNoLine{\@classname}{No country present for an affiliation}%
    \fi
}
\makeatother
\AtBeginDocument{%
  }

\setcopyright{acmlicensed}\acmConference[]{September}{2023}{New York, NY, USA}

\newcommand{\parab}[1]{\noindent\textbf{#1}}

\usepackage{amsthm, amssymb}
\usepackage{mathrsfs}
\usepackage{graphicx}
\usepackage[linesnumbered,ruled, vlined]{algorithm2e}
\usepackage{algorithmicx}
\usepackage{latexsym}
\usepackage{bm}
\usepackage{listings}
\usepackage{subfigure} 
\usepackage{threeparttable}
\usepackage{multirow}
\usepackage{pifont}
\usepackage{xcolor}
\definecolor{eclipseStrings}{RGB}{42,0.0,255}
\definecolor{eclipseKeywords}{RGB}{127,0,85}
\colorlet{numb}{magenta!60!black}

\definecolor{mygreen}{rgb}{0,0.6,0}
\definecolor{mymauve}{rgb}{0.58,0,0.82}

\definecolor{mygray}{gray}{.9}
\definecolor{mypink}{rgb}{.99,.91,.95}
\definecolor{mycyan}{cmyk}{.3,0,0,0}

\renewcommand\footnotetextcopyrightpermission[1]{}
\begin{document}

\title{ClickINC: In-network Computing as a Service in Heterogeneous Programmable Data-center Networks}


\author{Wenquan Xu$^\dagger$, Zijian Zhang$^\dagger$, Yong Feng$^\dagger$, Haoyu Song$^\star$, Zhikang Chen$^\dagger$, \\
Wenfei Wu$^{\mathsection}$, Guyue Liu$^{\ddagger}$, Yinchao Zhang$^\dagger$, Shuxin Liu$^\dagger$, Zerui Tian$^\dagger$, Bin Liu$^\dagger$}
\affiliation{%
  \institution{$^\dagger$Tsinghua University, $^\star$Futurewei, $^\mathsection$Peking University, $^\ddagger$New York University Shanghai}
  }
\authornote{Bin Liu, Wenfei Wu, and Guyue Liu are corresponding authors}

\renewcommand{\shortauthors}{Wenquan Xu et al.}

\begin{abstract}
In-Network Computing (INC) has found many applications for performance boosts or cost reduction. However, given heterogeneous devices, diverse applications, and multi-path network typologies, it is cumbersome and error-prone for application developers to effectively utilize the available network resources and gain predictable benefits without impeding normal network functions. Previous work is oriented to network operators more than application developers. We develop ClickINC to streamline the INC programming and deployment using a unified and automated workflow. ClickINC provides INC developers a modular programming abstractions, without concerning to the states of the devices and the network topology. We describe the ClickINC framework, model, language, workflow, and corresponding algorithms. Experiments on both an emulator and a prototype system demonstrate its feasibility and benefits. 
\end{abstract}

%

\maketitle

\section{Introduction}\label{sec:introduction}

Defying the conventional wisdom, network is no longer considered as dumb pipe but also a computation-facilitating infrastructure which can help boost application performance (e.g., latency and throughput) or reduce system cost (e.g., power and engaged servers). Such a paradigm shift, dubbed as \emph{In-Network Computing (INC)}, has benefited  many applications (e.g., key-value store~\cite{netcache2017, liu2017incbricks},
machine learning (ML) aggregation~\cite{aggregation2017, 265065, lao2021atp}, consensus~\cite{dang2015netpaxos,dang2016paxos},
coordination~\cite{netchain2018}, and streaming~\cite{streaming2018}). 
These applications are typically enabled by the programmable switches (e.g., Tofino~\cite{tofino2019}) which however is limited by hardware capability and capacity~\cite{kim2020tea}, arising a trend to extend on \emph{heterogeneous programmable network devices}~\cite{bressana2020trading,zeng2022tiara,sna_gw,ipdk} (e.g., Tiara~\cite{zeng2022tiara} achieves a layer-4 load balancer), where the switch is used to perform throughput-intensive task (packet encap/decap) and FPGA is used for memory-intensive task (physical server selection).

While this momentum is inspiring, a closer look reveals a less optimistic reality: the adoption of INC is currently limited to network operators and has not yet to be embraced by application developers, which hinders the development of new applications and their large-scale deployment. The fundamental reason, we believe, is the lack of a high-level programming framework that can abstract away the complexities associated with issues such as device heterogeneity, network topology, and function mapping. Early efforts~\cite{gao2020lyra,hoganmodular2022} attempted to improve the programming abstraction by hiding hardware details. Although this is a valuable first step, there are still three major barriers. To see why, consider the state-of-the-art framework Lyra~\cite{gao2020lyra}.

\parab{Limited to low-level abstractions.} Lyra progresses from low-level and chip-specific languages (e.g., P4~\cite{p4lang} and NPL \cite{Nplang}) to a more general and cross-platform language. However, it still requires programmers to handle low-level details such as packet header processing and network protocol handling, and is limited to basic statements (e.g., if-else), rather than more advanced features (e.g., for-loop). Crucial features such as network transparency, cross-device correctness, and program isolation are missing and need to be implemented by INC programmers. These burdens discourage application developers from adopting the INC programming paradigm. 

\parab{Limited to a small-scale deployment.} Lyra can run a data plane program on multiple heterogeneous ASICs in a distributed way (e.g., load balancer~\cite{ghomi2017load}, in-band network telemetry~\cite{int_spec}). It achieves this by encoding the logic and different resource constraints into a satisfiability modulo theories (SMT) problem, and using an SMT solver (e.g., Z3~\cite{moura2008z3} and cvc5~\cite{barbosa2022cvc5}) to find the deployment strategy. However, this approach is prohibitively slow (e.g., Z3 takes 30+ minutes to allocate ML Aggregation program on only 5 Tofino devices). Furthermore, it can only find a \emph{feasible} deployment without considering resource utilization, thus limiting it to running a small number of applications with fixed resources.

\parab{Limited to a single user.} Lyra, along with other prior work~\cite{sultana2021flightplan,hoganmodular2022}, is designed for network operators who have complete control over all network devices and run a \emph{monolithic} program in the target network. In case of any changes, the entire program must be recompiled from scratch and reinstalled on affected devices, leading to inefficiencies in terms of compilation and installation time. Furthermore, this approach is unsuitable for running programs from multiple users, as coordination between users is necessary and traffic from different users must be interrupted for every change.

Given these issues, we argue the need of a new framework that offers high-level abstractions for writing applications, and automatically handles low-level system concerns such as placement, cross-device communication, resource isolation, fault tolerance, and more. With such a framework, developers would be able to offload routine tasks to the framework and focus on the critical logic of their applications. 

In this paper, we present ClickINC, a framework for INC application developers (referred to as ``users'') to develop, deploy, and manage programs on heterogeneous programmable network devices in data centers.  At a high level, ClickINC offers the following capabilities: (i) ClickINC allows users to develop applications in a high-level, Python-style language; (ii) ClickINC's compiler frontend compiles each user's program into a platform-independent intermediate representation (IR) program and determines the optimal placement strategy across the network; then the backend translates IR programs into chip-specific programs and launch them on the target network; (iii) at runtime, ClickINC isolates resources for different users and allows for dynamically adding and removing programs. Compared to prior work, ClickINC makes the following three notable contributions:

\parab{1) Modular programming abstractions.}
ClickINC encapsulates common INC functionality into \emph{modules} such as various sketches, hash functions, providing users with a library. Users can work at a higher level of abstraction and use a simple Python-style syntax to import modules they need to write applications. This design eliminates the need for users to worry about the low-level details (e.g., packet-level processing and implementation of data structures), reducing the amount of code (at least 10 times lower), and enabling them to reuse code across multiple projects. The comparison between ClickINC and other operator-oriented languages such as Lyra can be drawn to that of Python and C/C++. While C/C++ is fast and efficient, it is suited for low-level system development, whereas Python is easier to learn and use, better suited for application development.

\parab{2) Scalable placement algorithm.}
Efficiently placing programs on a network of heterogeneous programmable devices is challenging. In addition to different hardware features and resource constraints considered by prior work~\cite{gao2020lyra,sultana2021flightplan}, we take into account three new factors: (i) the network may consist of multiple paths for an application; (ii) the interaction between program segments distributed across multiple devices may result in extra overhead; (iii) users may add or remove applications dynamically, but the placement recomputing from scratch should be avoided.
We propose a program partition theory and based on it, we develop a dynamic programming (DP) algorithm to solve the placement problem in polynomial time and scale up to \textasciitilde1,000 switches.

\parab{3) Incremental program compilation.} 
To effectively support the multi-user scenario where each user dynamically adds or removes a program, ClickINC provides the incremental program compilation feature. Unlike prior work, we need to consider not only the programs run by the operator for routine packet processing and forwarding, but also programs from users for high-level applications such as key-value stores. Our key idea is to maintain the operator's program as the \emph{base program}. By applying an annotation-based method, multiple user programs can be correctly identified and incrementally integrated to or stripped from the base program. 
When synthesizing the base program and multiple user programs, ClickINC isolates both programs' states and control flows, ensuring each user's traffic is processed by the corresponding program. 

We build an end-to-end system and implement three common INC applications. Our evaluations show that ClickINC is 10X better than the state-of-the-art programming languages in lines of code, as well as near 1000X faster than SMT solver for program placement, and 50\%-75\% less traffic is affected to deploy a new program.
\section{Background and Motivation} \label{sec:background}
We provide the background of INC in data centers, and then discuss the pain points of application developers to motivate the need of a new INC framework.

\subsection{INC in Data Centers} \label{ssec:problem-statement}
\parab{INC Applications.} 
We use two common applications, key-value store and ML gradient aggregation, to illustrate the process and benefits of adopting the INC paradigm.

\emph{(1) Key-Value Store (KVS).} Traditional KVS nodes are inefficient in handling skewed, dynamic workloads due to limited server performance. With the programmable network devices, the in-network KVS, typified by NetCache \cite{netcache2017}, can accelerate KVS with 3-10x throughput improvement and lower latency. To deploy KVS on a capability-limited switch, the advanced data structure on servers needs to be replaced with a hash based key-value table, a hit counter, and a combination of Count-Min Sketch and Bloom Filter, to support cache read/write and statistics of queries for cache update.

\emph{(2) ML Gradient Aggregation (MLAgg).} Traditional distributed MLAgg relies on a parameter server (PS) or allreduce, which have performance bottleneck on servers. To accelerate it, in-network aggregation~\cite{265065,
lao2021atp} maintains a stateful structure called aggregator array on switch to aggregate gradients from different workers, which greatly improves the aggregation throughput. Packets are addressed to aggregators by their job id and sequence number, and an aggregator sums up the data from all workers and returns the results back to workers. Due to the limited switch resource and capability, the data type conversion from floating-point to integer may be needed, and several other data structures are used to ensure correct aggregation.

\parab{Heterogeneous Programmable Devices.} 
The heterogeneous devices in DCN (e.g., switches, FPGA, and smartNIC) can be roughly classified as pipeline or multi-core devices. The former has a number of stages with each running  a piece of the program and can provide a throughput guarantee; the latter has multiple cores working in parallel and can support more complex functions. It may be infeasible to deploy an INC application on a single device or the same type of devices due to resource and feature constraints. For example, to aggregate the ML parameter with 64 integers in a packet, at least two Tofino switches are needed due to the limited on-switch memory. Further, if the parameters have large sparsity, the sparsity detection and elimination function cannot be run on the programmable switches and require another type of device (e.g., SmartNIC or FPGA).

\subsection{Pain Points for Application Developers} 
\label{sec:mov-pains}
We discussed the three problems
of INC program development in $\S$\ref{sec:introduction}, and we further
illustrate these problems using examples of state-of-the-art solutions.
\begin{figure}[t]  
	\centering  
	\includegraphics[width=0.49\textwidth]{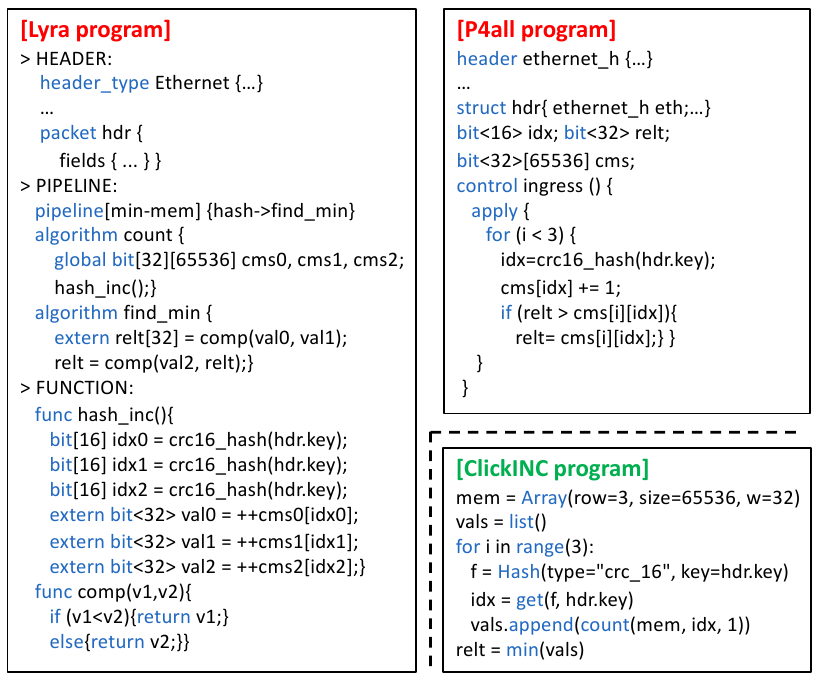} 
	\vspace{-0.4cm}
	\caption{Language comparison for count-min sketch.}
	\label{fig:lang_comp}  
	\vspace{-0.4cm}
\end{figure}

\parab{Low-level architecture and network details.} The existing INC programming is not friendly to application developers.
Fig.~\ref{fig:lang_comp} shows an example of the implementation of an in-network
count-min sketch in Lyra, P4all, and ClickINC. Lyra and P4all are network
operator-oriented and preserve device specific concepts such as pipeline, bit width,
and CRC. In contrast, ClickINC's basic programming elements are {\tt for} loop and Array,
which are organized following a Python-like high-level language syntax.
The ClickINC program is easier to learn and write, and needs fewer lines of code.

\parab{Limited number of devices and applications.} It is error-prone to place a program spanning multiple heterogeneous devices for applications with multi-path traffic. 
For example, to deploy an MLAgg on a fat-tree network, due to complex network topology and unbalanced resources, manual placement may cause: (1) some paths with an inadequate resource cannot be covered by MLAgg, so a lot of traffic cannot be aggregated; (2) cross-device interaction overhead as well as extra resource usage is high due to improper partition. SMT solver (used by prior work~\cite{gao2020lyra}) should have been a good tool to deal with this problem as the placement task can be modeled in SMT. However, such solvers need to traverse the entire solution space which has an exponential time complexity on both the number of instructions in the program and the number of devices.

\parab{Lack of user isolation.} Existing network devices do not provide isolation between user programs.
These devices were designed for a single-party operator, and thus do not have mechanisms (e.g., resource virtualization) to support multiple programs from different users.  
For example, if two users deploy the same Count-Min Sketch
program (Fig.~\ref{fig:lang_comp}) as two instances, with na\"ive program 
splicing, both users' traffic will be monitored at the same memory region {\tt vals.append()}. This may impact the accuracy of measurement and expose sensitive data (can be read by each other in {\tt relt}).

\section{ClickINC overview}
\label{sec:overview}
Our goal in designing ClickINC is to provide a framework for developing INC applications and automatically deploying them on heterogeneous programmable network devices in data centers. In practice, we also want (i) the developing environment to be friendly to developers, minimizing the effort required to apply the new INC programming paradigm, and (ii) the deployment to be compatible with existing INC deployments controlled by network operators. 

\subsection{Key Ideas} 
We first discuss the main ideas that enable ClickINC to tackle the three pain points discussed in~$\S$\ref{sec:mov-pains} while meeting the above two practical requirements.

\parab{1) A high-level, Python-style language with built-in modules.} We observed that the key obstacle to using the existing languages is it requires extensive architecture and network specific details. Inspired by the success of the high-level language Python, ClickINC provides users with Python-style syntax elements. Meanwhile, ClickINC encapsulates widely-used basic data structures (e.g., key-value matching table) and functions (e.g., hash) as ClickINC modules and builds a library for code reuse.
    
\parab{2) A scalable placement algorithm based on the program partition theory and topology compression.}  The large number of heterogeneous devices compounded with a substantial amount of program instructions makes the placement problem challenging. To reduce the search space, ClickINC merges dependent instructions into blocks to reduce the entities for placement, and leverages the symmetry of the fat-tree topology (the most common data center topology) to reduce the number of devices under placement consideration. Such optimizations enable our efficient DP algorithm to handle up to  \textasciitilde1,000 switches.

\parab{3) An annotation-based approach to incremental program compilation.} Running multiple programs from different users is challenging due to the potential resource conflicts; supporting \emph{dynamic} user requests which may add or remove a program is even harder. ClickINC enables dynamic user requests while accommodating existing operator programs by providing an incremental compilation feature. 
Our idea is to treat the operator's programs as \emph{base programs}. When synthesizing the base programs and user programs, ClickINC uses an \emph{annotation-based} approach to provide both memory isolation and control flow isolation. 

\begin{figure}[tb]
	\centering
	\includegraphics[width=0.7\columnwidth]{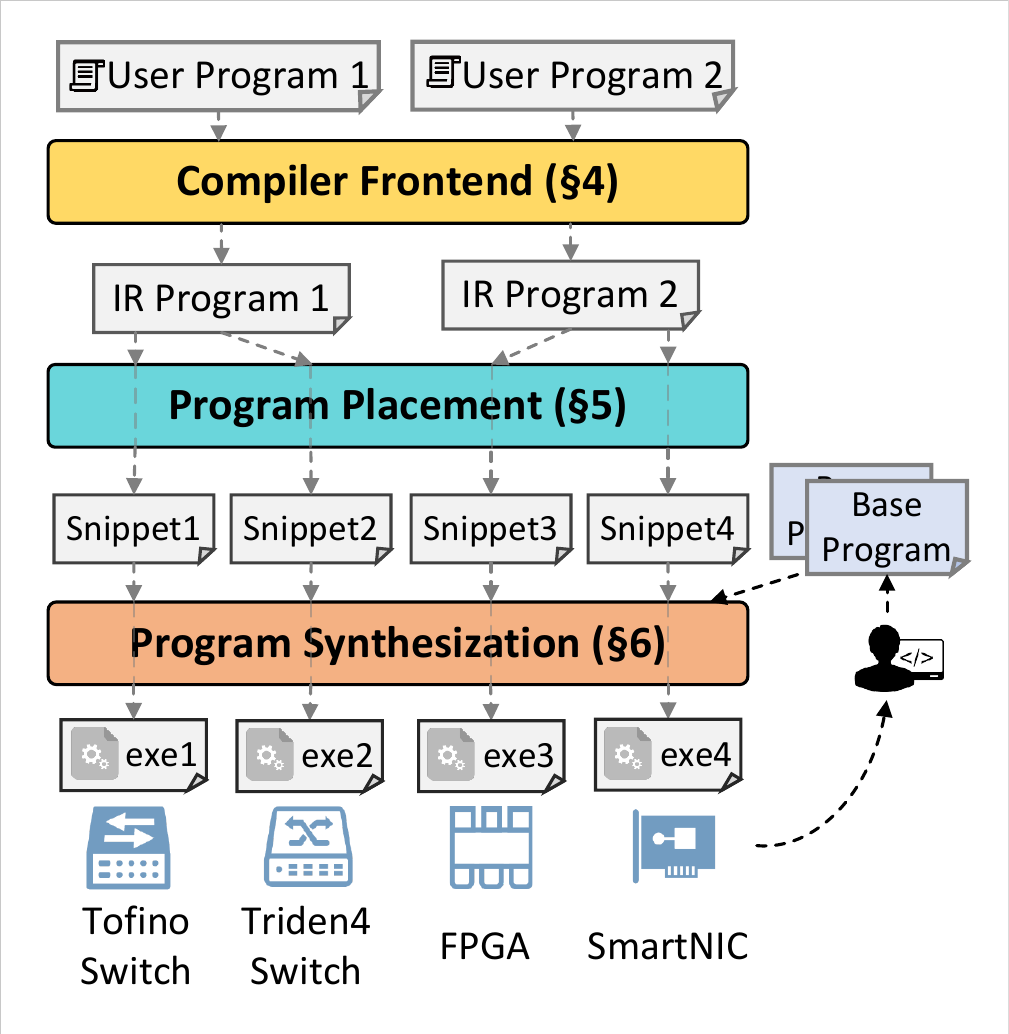} 	\caption{ClickINC architecture and workflow.} 
	\label{fig:architecture-workflow}
    \vspace{-0.2cm}
\end{figure}
\subsection{ClickINC's Workflow}
Fig.~\ref{fig:architecture-workflow} shows the overall architecture and workflow of ClickINC. At a high level, using ClickINC entails four steps: 

\noindent \textbf{(i) Writing a user program:} ClickINC provides users a high-level, Python-style language to write INC programs. 
Users can use \emph{built-in modules} which encapsulate common functions. Meanwhile, users can specify application performance requirements through the module parameters.

 \noindent \textbf{(ii) Compiling user programs to IR programs:} The ClickINC compiler frontend compiles each user's program into an Intermediate Representation (IR) program, where the IR instructions set is platform-independent. We choose the representative IR instructions from each platform and  merge the common ones.
 
 \noindent \textbf{(iii) Placing IR programs:} Then ClickINC decides a placement plan to deploy IR programs distributedly on network-wide heterogeneous devices. To handle a large number of programs and devices, ClickINC uses a dynamic programming algorithm to find the placement plan with the highest gain. Each user's program may be split into multiple \emph{snippets}, one for each device.

\noindent \textbf{(iv) Deploying on heterogeneous devices:} Finally, ClickINC compiler backend compiles snippets and the base programs (from the operator) to executable device programs in device-specific
languages. Each executable includes the base program and one or more user snippets running user-specific applications.
\section{Programming Abstractions}  \label{sec:programming-abstraction}
\label{sec:clickinc_api}

\subsection{User Programming}
\label{subsec:user_programming}

\parab{Abstraction/Interfaces.} INC programming is cumbersome. At the device
level, the heterogeneous resources, network topology, and target languages need to be considered; at the program level, a complete INC program needs to tend every packet handling detail including the inter and intra-device
communication protocol.  To hide the complexity, ClickINC is built on the One
Big INC~(OBI) abstraction (Fig.~\ref{fig:structure}) which contains elements in
three levels.
\begin{figure}[tb]
 	\centering
 	\includegraphics[width=0.49\textwidth]{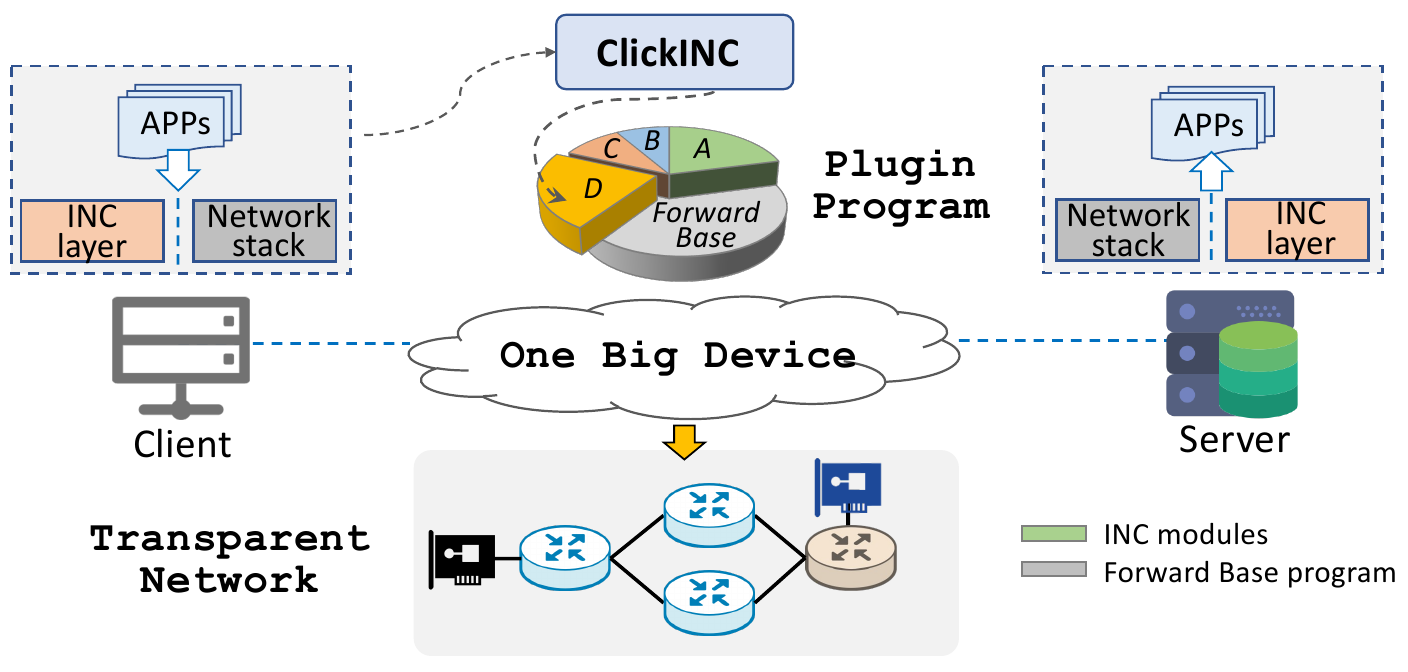}
 	\caption{ClickINC OBI Abstraction.}
 	\label{fig:structure}
 \end{figure}
\begin{figure}[tb]
\setlength{\belowcaptionskip}{-0.5cm}
	\centering
	\includegraphics[width=0.8\columnwidth]{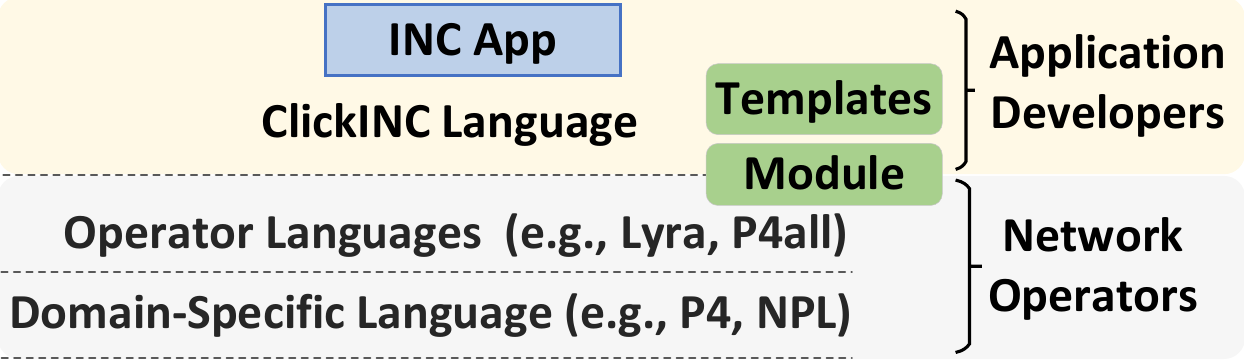}
	\caption{Languages to program network devices.} 
	\label{fig:programming-abstractions}
\end{figure}

{\textit{One Big Device.}} In OBI, the entire network is abstracted as
a single virtual programmable device $\mathcal{D}$ to the INC developers.  The
target devices comprise switch ASICs, multi-core smartNICs, FPGA smartNICs, and
FPGA accelerator card, denoted as $A$, $N_S$, $N_F$, and $F$, respectively.
Especially, a switch ASIC can be equipped with a bypass accelerator cards,
denoted as $\hat{A}$, to enhance its memory and processing capacity. Thus,
$\mathcal{D}=\{A, \hat{A}, N_S,  N_F, F\}$.

\textit{Transparent Network.} The above elements make an INC program
a piece of standalone software. Behind the scenes, packet modification on
devices (e.g., INC header insertion, removal, and update) is needed. ClickINC
handles all such works with a generic internal header structure by the ``INC layer'' that is maintained on each end device, and makes these issues transparent to both INC developers and end-host applications.

\textit{Plugin Program.} Although One Big Device frees developers
from dealing with device heterogeneity and network topology, an INC program
still needs to integrate with the underlying forwarding function and the
existing INC applications. OBI allows developers to focus on the INC function
alone and deem an INC program as a standalone plugin on the One Big Device. The
heavy lifting for program partition, mapping, and integration is handled behind
the scenes. An INC program can be plugged in or unplugged from the base
forwarding program without affecting existing INC functions.

\parab{ClickINC Language.} For best appeal to users, ClickINC preserves
the high-level language abstraction for application developers as illustrated in Fig.~\ref{fig:programming-abstractions}, which differs from operator language (e.g., Lyra) and domain-specific language (e.g., P4). Fig.~\ref{fig:language-grammar} shows
the grammar of the ClickINC language. A program consists of simple
and compound statements.  A simple statement can assign
an expression to a variable. A compound statement can control branching or looping.  
A branching statement is composed of a condition expression and two branch
bodies containing further statements. A loop statement is composed of a condition and a body to be executed if the condition is met. An expression is composed of basic 
operators (Python built-in) and operand 
(an expression, a variable, or a constant).  A
function is treated as an expression which outputs a result by computing on arguments.  ClickINC supports a Python-like coding style.

The ClickINC language introduces some INC specific elements to ease the programming on network devices. 
The \emph{Fields}, \emph{Objects}
and \emph{Primitives} abstractions are commonly used in INC
applications~\cite{netcache2017,tirmazi2020cheetah,lao2021atp,netchain2018}.  A
field is a data type that can be used to declare variables with the packet header semantic. An object is a collective data type used to
declare variable for five INC objects: Table, Array,
Seq, Hash, and Crypto. INC primitives, including Get, Write, Clear, Count, Drop, Fwd, and Copy, operates on the INC objects.

Each INC module is internally encoded in a platform-independent
language (i.e., the IR in $\S$\ref{ssec:program-ir}). When compiling user programs,
the ClickINC toolchain links the INC modules to their IR implementations.

\begin{figure}[t]
	\setlength{\tabcolsep}{1pt}
\centering
\small
	\begin{tabular}{rlll}
		\textbf{Program}	           &$G$&:== & $var$=$E$ $\big|$ $G$ $\big|$ if $C$: $G$ else: $G$ $\big|$ for $C$: $G$  \\ 
		\textbf{Predicate}	           &$C$&:== & $(E\&E)$ $\big|$ $(E|E)$ $\big|$ $\sim E$  \\ 
		\textbf{Expression}	           &$E$&:== & $V$ $\big|$ $var$ $\big|$ $const$ $\big|$ $F$ $\big|$ $E \odot E$\\ 
		\textbf{Function}	           &$F$&:== & max() $\big|$ min() $\big|$ range() $\big|$ slice() $\big|$ <{}< $\big|$ $\cdots$\\ 
		\underline{\textbf{Field}}	   &$V$&:== & \textbf{value} | \textbf{header}    \\ 
		\underline{\textbf{Object}}	   &$O$&:== & Table $\big|$ Array $\big|$ Hash $\big|$ Seq $\big|$ Sketch $\big|$ Crypto\\ 
		\underline{\textbf{Primitive}} &$P$&:== & get($O$) $\big|$ write($O$) $\big|$ clear($O$) $\big|$ count(O) $\big|$ \\
									   &   &    & del($O$) $\big|$ drop() $\big|$ fwd() $\big|$ copy($O$, $V$) \\ 
	\end{tabular}
 \vspace{-0.3cm}
	\caption{ClickINC grammar. $\odot$ denotes arithmetic or bit operations, and underlined elements are ClickINC specific (see Table~\ref{tab:function} in Appendix~\ref{appendix:pyinc} for ``Function F'').}
	\label{fig:language-grammar} 
    \vspace{-0.6cm}
\end{figure}
\parab{Modular Programming.}
The INC service provider implements the INC specific elements
as modules. 
With such modular programming, we incorporate the INC-related 
data structures and operations into a user-friendly high-level programming environment.
A user can assemble a program with the ClickINC language and the INC modules. Fig.~\ref{fig:lang_comp} shows an example of implementing a Count-Min Sketch using the INC object Array and Hash function.

\parab{Template.} The service provider can also
define common INC programs as \emph{templates}, and provide them to
users as libraries. ClickINC provides the templates for MLAgg, KVS,
and DQAcc (for SQL DISTINCT function), which cover a broad range of INC applications. 

To use a template, users need to provide a configuration profile, so that to configure the module/template parameters. A configuration profile includes the four fields: \emph{App}: dedicated ID of template; \emph{Performance}: performance requirements provided by users; \emph{Traffic distribution}: the upper limit of the querying frequency of clients; \emph{Packet format}: the format of packet header of applicaition. Fig.~\ref{fig:configuration-kvs} shows an example profile for KVS template.
Users can configure module/template data structures directly, e.g., Array size. Certain modules may need hardware-specific configurations that are obscure to users. In
this case, ClickINC provides the application performance metric for the user; For example, Fig.~\ref{fig:configuration-kvs} shows that a key-value search user requires a KVS module with the hit ratio and the accuracy of statistics for missed queries, weight by 0.7 and 0.3 respectively. Especially, as the OBI abstraction makes device transparent to users, leading to the difficulty of setting resource-related parameters. Therefore, ClickINC pre-learns a model to automatically set parameters based on empirical experimental results.
The details of the templates and their configuration can be found in Appendix~\ref{appendix:pyinc}.

Moreover, users can also incrementally add new logic to the existing templates, saving
the efforts to ``re-invent wheels''. For example, Fig.~\ref{fig:program-sparse-agg}
shows how a user can build a customized sparse gradient aggregation
based on the MLAgg template: The user program first imports and customizes a MLAgg template as an instance (line 1); then detects the sparse part of the
parameter vector and drops the sparse one (line 5-9); only the dense one will be aggregated by MLAgg instance (line 10).

\lstdefinelanguage{json}{
    basicstyle=\normalfont\ttfamily,
    commentstyle=\color{eclipseStrings}, 
    stringstyle=\color{eclipseKeywords}, 
    numbers=left,
    numberstyle=\scriptsize,
    stepnumber=1,
    numbersep=8pt,
    showstringspaces=false,
    breaklines=true,
    frame=lines,
    string=[s]{"}{"},
    comment=[l]{:\ "},
    morecomment=[l]{:"},
    literate=
        *{0}{{{\color{numb}0}}}{1}
         {1}{{{\color{numb}1}}}{1}
         {2}{{{\color{numb}2}}}{1}
         {3}{{{\color{numb}3}}}{1}
         {4}{{{\color{numb}4}}}{1}
         {5}{{{\color{numb}5}}}{1}
         {6}{{{\color{numb}6}}}{1}
         {7}{{{\color{numb}7}}}{1}
         {8}{{{\color{numb}8}}}{1}
         {9}{{{\color{numb}9}}}{1}
}

\parab{User-defined Module.} Although we suggest the modules to be implemented by the service provider for simplicity, ClickINC 
reserves the flexibility for users to design their own INC modules, called user-defined modules. 

To develop a user-defined INC module (i.e., object and primitive shown in
Fig.~\ref{fig:language-grammar}), a user needs to use the ``low-level''
instructions to write the module program. These low-level instructions could be
IR instructions or operator-level instructions (like Lyra, and shown in
Table~\ref{tab:atomic}$-$Appendix~\ref{appendix:IR}).

\begin{figure}[tb]
\begin{lstlisting}[language=json, basicstyle=\ttfamily\footnotesize]
{"app" : "KVS",
 "performance":
 {"objective function": max 0.7hit+0.3acc,
  "content": >=1000, ...},
  "traffic frequency": {c1: 10Mpps, c2: 20Mpps, ...},
  "packet_format":
  {"network": "ethernet/ipv4/udp",
   "khdr": {"key": "bit_128"},
   "vhdr": {"value_0": "bit_32"}, ...
} }
\end{lstlisting}
	\caption{Configuration for KVS}
	\label{fig:configuration-kvs}
\end{figure}
\begin{figure}[tb]
\begin{lstlisting}[language=python] % \footnotesize
from Template import MLAgg
agg = MLAgg(row, dim, is_convert, scale)
for i in range(BlockNum):
  sparse = 1
  for j in range(BlockSize):
    index = BlockNum * i + j
    if hdr.feat[index]!=0:
      sparse = 0
  if sparse = 0:
    del(hdr.feat[index])
agg(hdr)
\end{lstlisting}
\vspace{-0.3cm}
	\caption{The user program based on the MLAgg template, performing sparse gradient aggregation.}
	\label{fig:program-sparse-agg}
    \vspace{-0.3cm}
\end{figure}
\subsection{Program Intermediate Representation} \label{ssec:program-ir}
\parab{Platform-Independent Intermediate Representation.} To compile a user
program to machine code on heterogeneous devices, ClickINC first compiles the
program into an Intermediate Representation (IR), called an IR program.

ClickINC summarizes the IR instruction
set from the different platforms it supports. The IR instruction set is listed in
Fig.~\ref{fig:IR-instruciton-set} in Appendix~\ref{appendix:IR}. Some instructions
are common on all platforms and the others only run on certain platforms: there are 13 classes of them, and
each platform supports a subset of them as shown in Table~\ref{tab:prmitive_block} in
Appendix~\ref{appendix:IR}. Such instruction constraints will take effect in later program placement.

The ClickINC IR instruction sets include declaration instructions and operation instructions,
where the former defines variables and the latter operates on variables.
As the ClickINC IR instruction set needs to adapt to programmable network devices (e.g., in pipeline switches, a packet must sequentially traverse the pipeline stages without
rewinding back in one pass), it does not support control flow transition (i.e., instructions like {\tt goto} or {\tt jump}). An IR program is therefore executed sequentially. 

\parab{Compiler Frontend.} The ClickINC frontend compiler compiles a user program into
an IR program in the following passes: (1) inline all the bodies of the functions in the main program from the unified library; (2) unroll the loops if it makes constant pass of the iteration, e.g., {\tt for i in range(3)}, otherwise an error will be reported; (3) convert the if-else branches to ternary operators in the format of \texttt{condition? instr}; (4) split the instructions into single-operand ones.
Especially, instructions with temporary variables are transformed into Static Single Assignment (SSA) to eliminate the write-after-read and write-after-write dependencies, helping the
IR DAG construction in later program placement ($\S$\ref{ssec:block-construction}). 

\section{Program Placement}
\label{sec:program-placement}

We formulate the problem of distributing an IR program on multiple devices
as an optimization problem and solve it by a dynamic programming algorithm.

\subsection{Problem Statement} \label{ssec:problem-statement}
The network contains multiple programmable devices, and a program's instructions
can be distributed on multiple devices.
Placing an IR program on network devices is an \emph{optimization problem}, where we
wish to maximize the traffic volume served by INC while minimizing the resource
consumption as well as the inter-device communication overhead. 
The solution needs to make a few tradeoffs: placing more blocks on a device is simpler but limits the INC capability and capacity, and distributing blocks on multiple devices incurs inter-device communication overhead.
Another key invariant in program placement is to keep the program
execution equivalent as being executed on a single device.

\parab{Complexity.}
To place a program on multiple devices is to find devices for each instruction
of the program. The problem's searching space exponential of program
size and network scope, i.e., $O((M\cdot S)^N)$ where $M$, $S$, and $N$ are the
number of devices, pipeline stages or cores for SoC device, and IR
instructions.  Na\"ive methods may find sub-optimal results: greedily choosing
a single path cannot utilize the multi-path resources; simply replicating
the program on all paths could lead to device overloaded. Existing methods usually cannot handle the problem in a large scale. Lyra needs manually labeling the candidate devices of a program, which
limits the result optimality; if ``all'' devices are labeled as candidate
devices, its SMT solver approach cannot give the result in an acceptable time.

\parab{Intuition.} We take three intuitions to reduce the algorithm complexity in ClickINC.
First, we group IR program instructions into \emph{blocks},
where all instructions in a block are executed all or none, and thus, one block
can represent all its instructions in the algorithm (reducing $N$,
$\S$\ref{ssec:block-construction}).  In a DCN, there could be multiple paths
between two communicating INC hosts. On each path, the IR program blocks must
be placed sequentially; among the paths, blocks are replicated on devices to
guarantee the traffic on different paths is processed by the same program;
two paths' intersection segment can hold blocks shared by both paths.

Second, we group DCN devices into equivalent classes, and use a class to
represent all its devices in the algorithm (reducing $M$).  
Third, we find that the placement problem can be
divided into isomorphic sub-problems, and thus propose a dynamic programming
(DP) algorithm to search for the optimal solution, which gives a solution in
polynomial time.

\begin{figure}[t]  
	\centering  
	\includegraphics[width=0.49\textwidth]{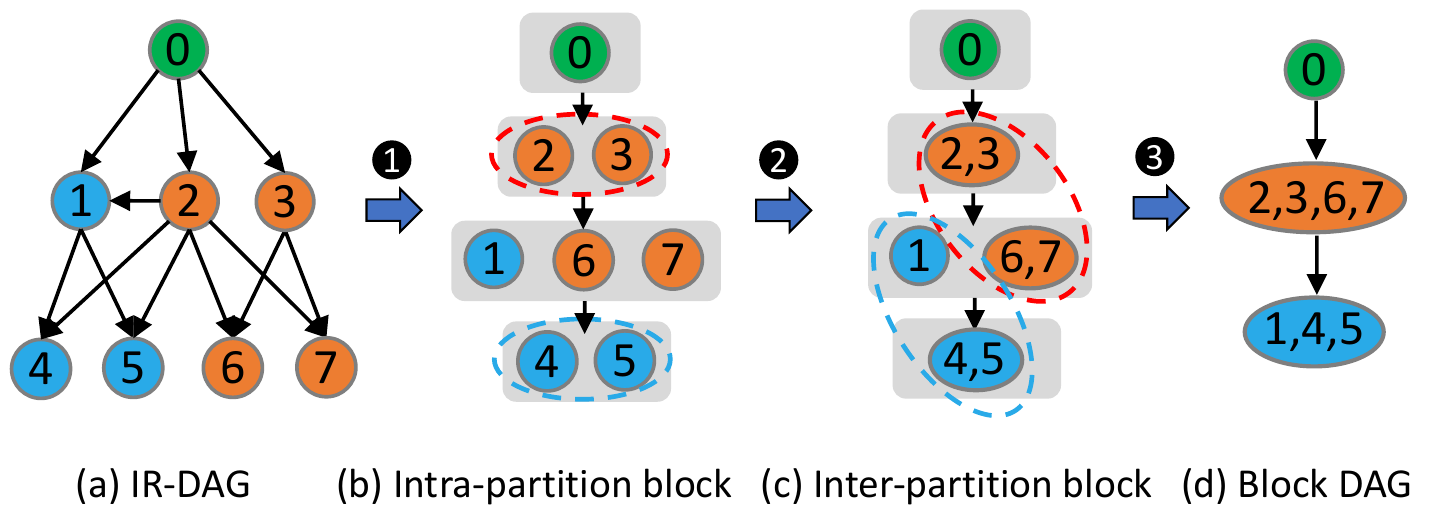} 
	\vspace{-0.4cm}
	\caption{An example of block construction.}
	\label{fig:block_construction}  
	\vspace{-0.5cm}
\end{figure}
\subsection{IR Block DAG Construction} \label{ssec:block-construction}
ClickINC first transforms
the IR program into a Directed Acyclic Graph (DAG) of disjoint instruction blocks to
comply with the sequential instruction execution. A block is a basic placement unit. Each block contains instructions in the original order as in the IR program, and
the union of blocks equals the IR program.  

In ClickINC, the IR block DAG construction should also comply with several
practical principles to ensure correctness. First, the instructions
operating on the same state should be in the same block to avoid
inconsistency. Second, the instructions in the same block should be of the
same type to ensure the block can be placed on some devices (not all devices support all instruction types). Third, a
block's size should be limited by a threshold parameter decided by the device capability.  Appendix~\ref{sec:partition_theory} formalizes these principles.

ClickINC initializes each instruction as a block and gradually merges the
blocks complying with the above constraints.  The algorithm takes three steps.

\parab{Step 1: construct instruction dependency graph.} If an instruction $i$ reads a 
variable whose value is written by a previous instruction $j$, $i$ \emph{depends on}
$j$. INC applications have a subtle pitfall: the program is driven
by packet arrival events and there are inter-packet \emph{states} (e.g., a packet counter). \emph{All instructions that write or read the same state are mutually dependent.} The other variables with a life span of a packet are called \emph{temporary variables}.

\parab{Step 2: merge instructions within a loop.} The IR program can be viewed as
a directed graph $G$, with the instruction as the node $V$ and the dependency
as the edge $E$. ClickINC iteratively merges nodes that form a loop. When multiple nodes
(denoted as $N$) are merged as one, a new node (i.e., block) forms to replaced the old ones, and the edges
between the merged nodes $N$ and the other nodes $V-N$ are replaced by edges
between the new node and other nodes. The algorithm repeats until there is no loop in the graph.

\parab{Step 3: merge non-exclusive blocks to compact the DAG.} After
eliminating loops, the graph becomes a DAG. ClickINC further runs Kahn's
topological sort algorithm
to partition the graph and merges non-exclusive blocks.
Fig.~\ref{fig:block_construction} illustrates an example. Kahn's algorithm
takes iterations to partition a DAG: each iteration takes the nodes whose input
degree is 0 as one partition and removes these nodes and their related edges
(Fig.~\ref{fig:block_construction}b).  After the Kahn partition, ClickINC
further merges blocks whose instructions are of the same type, i.e., being non-exclusive, within the same
partition (Fig.~\ref{fig:block_construction}b-c) and the adjacent partitions
(Fig.~\ref{fig:block_construction}c-d) without exceeding the block size limitation. 
The process repeats until no more blocks
can be merged.

\subsection{Topology Simplification}\label{ssec:topology-reduction}
\begin{figure}[tb]  
	\centering  
	\includegraphics[width=0.48\textwidth]{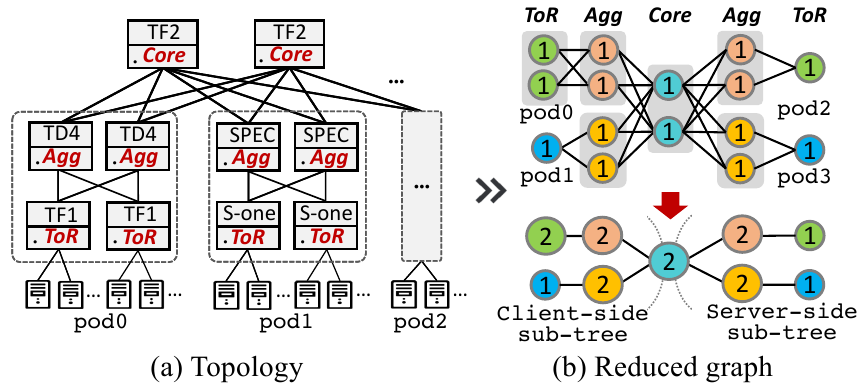} 
	\vspace{-0.6cm}
	\caption{Topology Simplification (number in a circle: the number of merged devices, color: device type).}
	\label{fig:multi-path} 
	\vspace{-0.6cm}
\end{figure}

ClickINC further reduces the search space for program placement by simplifying 
the network topology. ClickINC leverages the DCN's topological characteristics
to make the reduction. The network devices in a DCN can be divided into
several \emph{equivalent classes (EC)}, where devices in the same class have the 
same physical wiring with the other classes.

For a three-tier fat-tree topology,
all its ECs can be computed bottom-up in the topology.  
All ToR switches connecting with the same servers form an EC, all
aggregation switches connecting with the same ToR switches form an EC, and all
core switches form an EC (as they connects to the same aggregation switches). Based on the proof of the device equality in EC for program placement (see Appendix~\ref{appendix:proof_reduced}), we can
merge the switches in an EC as one virtual node, and thus the DCN topology is simplified to a tree (Fig.~\ref{fig:multi-path}).

In the later program placement, ClickINC also takes advantage of the path symmetry.
All physical servers are at the leaf nodes of the topology, and traffic goes upwards to 
a root and goes downwards along the tree. Thus the tree is segmented into two parts by the root, i.e., the client-side sub-tree and server-side sub-tree.
\subsection{Placement Algorithm} \label{ssec:program-placement}

\textbf{Optimization Goal.} The program placement algorithm aims to find a solution to maximize the traffic served by INC
with the minimum resource consumption and network bandwidth for passing parameters
between blocks ($\S$\ref{ssec:refine-data-plane}). 
With $x_{v,d}\in\{0, 1\}$ indicating whether block $v$ is placed
on device $d$, the objective $G(x)$ can be formalized as: 
\begin{equation}
\label{eq:op_target}\small
    G(x) = \omega_t h_t(x) - \omega_r h_r(x) - \omega_p h_p(x),
\end{equation}
where $h_t$ is the ratio of traffic served by INC, $h_r$ is the ratio of
resource consumed on devices, and $h_p$ is the ratio of data transferred across
devices. The parameters $\omega_t$, $\omega_r$, and $\omega_p$ balance the
three factors.
We empirically set $\omega_t$ as 1/2 to prefer high throughput, and
tune $\omega_r$ and $\omega_p$ dynamically according to the resource availability as the algorithm proceeds.

\begin{small}
\begin{equation}
    h_t(x)=\sum\nolimits_{l\in L_p} \left(\bigwedge\nolimits_{v\in P}\sum\nolimits_{d\in l}x_{v,d}\right) \times \frac{t_l}{\sum_{l\in L_p} t_l}, \nonumber
\end{equation}
\end{small}
i.e., the overall normalized traffic volume on the selected paths, where $t_l$ is the traffic volume on each path;
\begin{small}
\begin{equation}
    h_r(x)=\sum\nolimits_{d \in D} \sum\nolimits_{v\in V} x_{v,d}\times \frac{r(v)}{\sum\nolimits_{v\in V}r_v}, \nonumber
\end{equation}
\end{small}
i.e., the overall normalized resources on the selected devices;
\begin{small}
\begin{equation}
    h_p(x)=\sum\nolimits_{d_i, d_j \in D}\sum\nolimits_{v_k, v_l \in V}x_{v_k, d_i}x_{v_l, d_j}\times \frac{\phi_{v_k,v_l}}{\sum\nolimits_{d\in
D}\sum_{v\in V} x_{v,d}\phi(v)}, \nonumber
\end{equation}
\end{small}
i.e., the overall normalized volume of extra parameter incurred due to program
partition between selected devices, where $\phi_{v_k,v_l}$ denotes the amount
of extra data transferred between devices $v_k$ and $v_l$, and $\phi_{v}$ refers to all extra data
incurred by the block $v$.

\parab{Dynamic Programming Algorithm for Placement.}
Even with the reduced topology, the searching space to find an optimal placement of IR program is still too large due to the possible multiple flow paths from multiple pods. SMT or ILP solvers cannot give the solution in an acceptable time.
ClickINC uses an innovative dynamic programming algorithm with pruning. In detail, for the two sub-trees illustrated in Fig.~\ref{fig:multi-path}, we try to allocate the program but from different directions (i.e., sequentially allocate instruction blocks from leaf to root for the client-side sub-tree and do it in the reverse order for the server-side sub-tree, so that the problem is translated into two sub-tree-based program placement). Then we link the two sub-tree placement results by the root node, i.e., traverse all partial placement results of sub-trees, and choose the one with the largest gain of Eq.~\ref{eq:op_target} from all feasible combinations. 
The placement task on each sub-tree devices can be discomposed as two sub-tasks: (1) place the instruction blocks across devices for multi-path traffic; (2) decide the placement of instructions in each block within a particular device. We illustrate how ClickINC addresses these two sub-tasks.

\noindent{$\bullet$ \emph{Cross-device multi-path solution.}} Let $H_{B, D_i}$ denote the 
maximum gain of placing block(s) $B$ on a tree with $D_i$ as the root. 
When the tree is a single device $D_i$, $H_{B, D_i}$ equals the gain of Eq.~\ref{eq:op_target}.
When the tree has subtrees, ClickINC places a partition $B'$ on the root node (the partition can be $\emptyset$ or $B$)
and the remaining onto the subtrees, and the gain $H_{B, D_i}$ is the sum of
that on the root and the subtrees; by iterating all possible partitions,
ClickINC finds the one which gives the maximum $H_{B, D_i}$, i.e., 
\begin{small}
\begin{equation}
\label{eq:dp_cross}
    H_{B,D_i} = 
		\max_{B' \in Partion(B)}\left(\sum\nolimits_{j\in son(D_i)} H_{B-B',D_{j}} + G(D_i, B') \right).
\end{equation} 
\end{small}
The problem can be recursively divided into isomorphic sub-problems. We design
a dynamic programming algorithm to compute the problem bottom-up; the pseudo-code
is shown in Algorithm \ref{alg:dp}: line 1 adjusts weights; 2-3
uses Depth First Search (DFS) to traverse two sub-trees and performs
allocation; then for a leaf node, line 20-21 enumerates instruction blocks
and calls Algorithm \ref{alg:dp_within} to place instruction in blocks within a
device; for internal nodes, line 16-17 integrates allocation results of
possible branches and line 22 executes the DP following Eq.~\ref{eq:dp_cross}
where $calc\_hp(\cdot)$ computes the cross-device communication overhead.
Especially, as illustrated in line 10, we prune the illegal enumeration results
that violate block dependency to reduce the solution space. This algorithm can
be applied to fat-tree and spine-leaf topology with any number of layers.
\begin{algorithm}[t]
	\caption{Multi-path allocation}\label{alg:dp} \footnotesize
	\KwIn {$R$, $S$, $D$: the set of resources, stages, all available devices, $\mathcal{B}$: the set of instruction block to be allocated.}
	\KwOut {$s$: the allocation solution.}
	$\omega_t,\omega_p, \omega_r \leftarrow$ adjust($R$, $D$, $S$);\\
	CDP$\leftarrow$\texttt{DFS\_DP}(CTree.root, 1); \\ 
	SDP$\leftarrow$\texttt{DFS\_DP}(STree.root, -1);\\
	\For{B $\in$ CDP[CTree.root]}{
	    B' $\leftarrow \mathcal{B}$-B;\\
	    \If{B'$\in$ SDP[STree.root]}{
	        s $\leftarrow$ $\min$(s, CDP[CTree.root]+SDP[STree.root])
	    }
	}
    \textbf{return} s; \\
    \SetKwFunction{FMain}{DFS\_DP}
    \SetKwProg{Fn}{Function}{:}{}
    \Fn{\FMain{r,d}}{
    A $\leftarrow \varnothing$;\\
    \If{r = $\varnothing$}{
        return
    }
    \For{c $\in$ r.child}{
        $DP_{sub}$[c] $\leftarrow$ \texttt{DFS\_DP}(r)
    }
    sub\_G[$\varnothing$] $\leftarrow$ 0;\\ 
    \For{i $\in$ $\cup$ $DP_{sub}$[r.child].keys}{
        sub\_G[i] = sum($DP_{sub}$[i])
    }
    \For{i $\in$ sub\_G.keys}{
        $B_{ava} \leftarrow \{b|b\in \mathcal{B}-i; in\_degree(b,d)=0 $\};\\
        \For{B $\in$ \textbf{enum} $B_{ava}$}{
            curr $\leftarrow$ call $Algorithm$ \ref{alg:dp_within}(S[r], R[r], B);\\
            DP[r][i+B] $\leftarrow \max$(DP[r][i+B], sub\_G[i]+curr+$calc\_hp$(i,B));\\
        }
        }
    }
    \textbf{return} DP; \\
\end{algorithm}

\begin{algorithm}[t]
	\caption{Instruction allocation within a device}\label{alg:dp_within} \footnotesize
	\KwIn {$S_d$, $R_d$: the stages, resources of device $d$, $P$: set of instructions to be placed.}
	\KwOut {$I$: Instruction allocation results.}
	$I$[-1] $\leftarrow \{\varnothing:0\}$;\\ 
    \For{s$\leftarrow$ 0 \textbf{to} $S_d$}{
        \For{i $\in I[s-1]$}{ 
            \If{calc\_resource(p)$\leq R_d[s]$}{
                $I[s][i]$ $\leftarrow \max(I[s][i], I[s-1][i])$;\\
                $P_{nd} \leftarrow \{p|p\in P-i; in\_degree(p)=0\}$;\\ 
                \For{p $\subseteq$ \textbf{enum} $P_{nd}$}{
                    \If{$\exists i'\in I$[s].keys $\&\&$ i+p $\subseteq$ i'}{
                        continue;\\
                     }
                    \If{$\exists i'\in I$[s].keys $\&\&$ i' $\subseteq$ i+p}{
                        del I[s][i'];\\
                     } 
                    $I$[s][i+p] $\leftarrow\max$($I$[s][i+p], $I$[s-1][i]+$G$(p));\\
                }
            }
        }
    }
    \textbf{return} $I$;\\
\end{algorithm}

\noindent{$\bullet$ \emph{Intra-device solution.}}
To place instructions within a device, we use another DP algorithm to ensure (1) the instructions satisfy resource constraints; and (2) the placement has the largest gain according to Eq.~\ref{eq:op_target}. Thus, we can derive:
\begin{small}
\begin{equation}
	H_{p,S_i} = \max_{p' \in Partition(p)}\left( H_{p-p',S_{i-1}} + G(S_i, p')  \right),
\end{equation}
\end{small}
where $p$ is the instructions that are placed, $S_i=[s_1, s_2, \cdots, s_i]$ is the set of stages for pipeline devices ($S_i=s_0$ for non-pipeline device). On a pipeline, the instruction-to-stage mapping has a huge solution space. To improve efficiency, \ding{182} the infeasible solutions violating the instruction dependency are pruned (line 6 of Algorithm 3); \ding{183} the target function Eq.~\ref{eq:op_target} prefers solutions with more compact placement (i.e., each stage should use up at least one type of resource), so inadequate solutions are pruned (line 8-12). With these, the DP algorithm achieves the similar solution as SMT in a much shorter time. The pseudo-code is shown in Algorithm \ref{alg:dp_within}. 

\parab{Adaptive Weight.}
As the algorithm proceeds, $\omega_r$ is set as $\omega_r=1-2^{r-1}$ and
$\omega_p=1/2-\omega_r$, where $r$ is the ratio of remaining resources. The
adaptive weight could raise the importance of device resource
allocation as the remaining resource decreases (a smaller $r$ leading to a
larger $\omega_r$).

\parab{Placement Constraints and Pruning.} As the DP algorithm searches the solution with the highest gain, the following 
pruning techniques are applied to reduce the search space. When one of the following constraints is violated, the algorithm sets 
$H_{B,D_i}$ as negative infinity ($-\infty$) and 
stops exploring the branch: 
(1) if a device's resource capacity cannot satisfy the 
block; (2) if an instruction placement violates the instruction dependency; (3) if a device's computation capability fails to satisfy the block's instruction type. Besides, the target function Eq.~\ref{eq:op_target} prefers solutions with more compact placement (i.e., each stage should use up at least one type of resource), so inadequate solutions are pruned.

To map the program on the devices with various constraints, we propose device modeling based on different architectures (i.e., pipeline, multi-core) to formalize the device-level instruction placement in Appendix~\ref{appendix:device_modeling}, and describe the chip-specific constraints in Appendix~\ref{appendix:contraints}.
\section{Program Synthesis}
\label{sec:program-synthesization}

Each device runs a network operator-deployed program, called \emph{base program}, 
to perform the basic network functions such as packet validation, forwarding, etc.
Multiple users' INC programs (snippets) placed on the device rely on ClickINC to synthesize them as one big program.

A program typically consists of a header parsing snippet and a packet processing snippet.
User programs and the base program could parse different header fields for their own packet processing.

\parab{Refine Runtime Data Plane.} \label{ssec:refine-data-plane}
The network data plane are refined to support 
program execution on distributed devices. 
The two refinements are transparent to the users: for a user's traffic,
the first network device inserts a special header for the following refinement,
and the last one removes it.

First, 
temporary variables may be shared by multiple devices.
ClickINC allows the user packets to carry the shared
variables from one device to its downstream devices. 
ClickINC packet header has a field {\tt Param} to store the 
temporary variables. Note that persistent variables are only 
used and placed on one device, and the static single assignment
transformation makes temporary variables only have dependency 
from the successor device to the predecessor along the DAG.

Second, ClickINC allows placing replicated blocks along a path.  For
example, a program with blocks 1, 2, and 3 may be placed along a four-hop path
as 1, 2, 2, 3.  Thus, ClickINC needs to decide and tell the devices with
replicated blocks which of them processes a packet. ClickINC assigns each block in the DAG
program a step number, and adds a {\tt step} field in the packet header. A
device attempts to match the packet {\tt step} field with its own block's step, if they
match, the block is executed and the packet step is increased to the next step,
otherwise, the packet skips the processing (if the packet step number is
larger) or dropped (if the packet step number is smaller).
Allowing replicated blocks in the network also provides another advantage: if the 
network experiences a transient failure, a packet can skip the faulty device
and get processed by the successor device with replicated blocks.

\parab{Compiler Backend.} ClickINC first isolates user programs from each other and the base program.
It renames variable in the user programs, so that 
after compilation their programs access isolated memory region, without violating each others' data.
For example, the {\tt mtb} variable in a KVS program {\tt kvs\_0} is renamed as {\tt kvs\_0\_mtb}.
Then it adds a user ID match to filter out the user's traffic for its own program.
\begin{lstlisting}[language=c++]
if (INC_1_hdr.isValid()) {logic1;}
\end{lstlisting}

ClickINC compiles each program individually into device-specific instructions,
called \emph{device program}. These device programs are merged with the following
optimization, and eventually compiled as an executable.

\begin{figure}[t]  
	\centering  
	\includegraphics[width=0.46\textwidth]{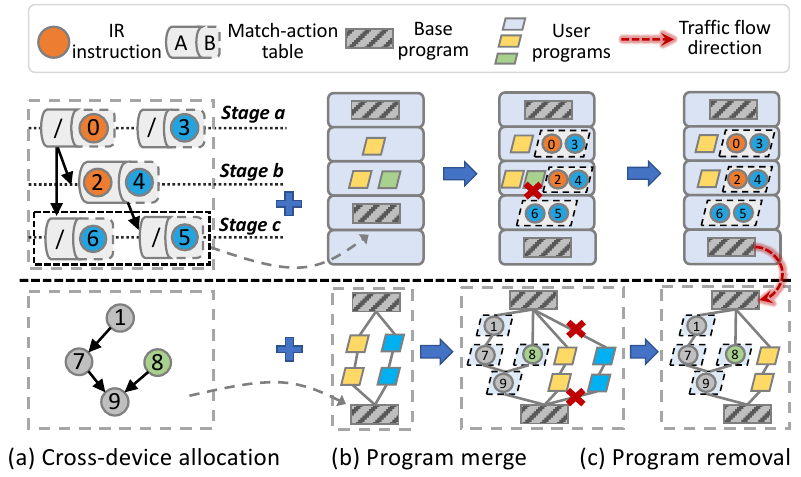} 
	\vspace{-0.4cm}
	\caption{Program Synthesis.}
	\label{fig:IR_pdg}  
	\vspace{-0.5cm}
\end{figure} 

\parab{Program Merge.} ClickINC merges header parsing snippets and packet
processing snippets separately. The header parsing follows a tree
structure. When merging two programs’ header parsing, ClickINC scans both trees, merges the different branches, and eventually outputs a merged tree.

Merging packet processing snippets is  more complex due to the dependency between the user programs and the base program. For example, the forwarding function in the base program depends on the user program if the user program changes the packet's IP addresses (e.g., NetCache~\cite{netcache2017});
the user programs depend on the  packet integrity check function in the base program, because 
only valid packets should be handed to the user programs. Thus, the base program is
divided into a \emph{head} part and a \emph{tail} part, where \emph{head} contains functions depended on
by the user programs and \emph{tail} contains functions depending on the user programs.

For pipeline devices, as the upper part of Fig.~\ref{fig:IR_pdg}(b) shows, the user program is placed between \emph{head} and \emph{tail} of the base program. The user program is moved to stages as early as possible to reduce the overall stages. For multi-core devices, ClickINC merges the dependency graphs of the user program and the base program according to node dependency, and then merges the corresponding code pieces based on the topological sorting order on the merged graph, as illustrated in the lower part of Fig.~\ref{fig:IR_pdg}(b).

\begin{figure}[t]  
	\centering  
	\includegraphics[width=0.47\textwidth]{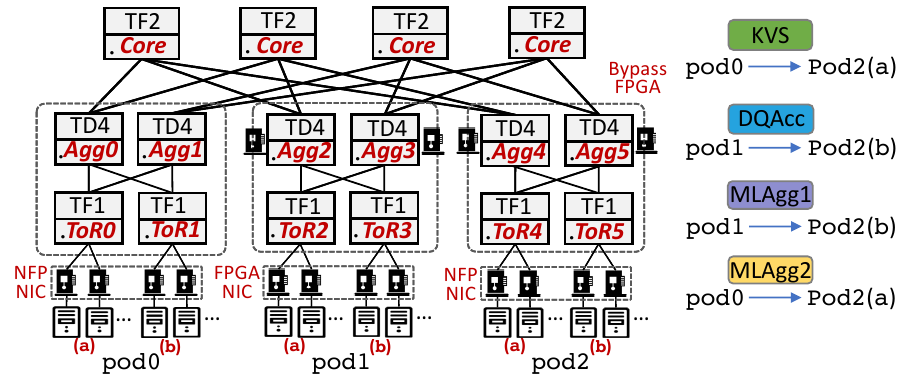} 
	\vspace{-0.3cm}
	\caption{Network Topology in Emulation.}
	\label{fig:emulation}  
	\vspace{-0.15cm}
\end{figure}
\parab{Incremental Compilation for Dynamic Program Merge \& Removal.}  
ClickINC applies an annotation-based method to support incremental user program merging and removal.
ClickINC associates each user program with an annotation indicating its ownership.
During the compilation, the annotation is associated with each instruction.
When merging a user program into the base program, ClickINC incrementally adds the new user annotation
to the shared instructions and sets the new user's own instructions with its annotation.

When a user revokes its INC service request, ClickINC iterates the synthetic program's instructions,
and removes the user's annotation; if an instruction has no annotation, the instruction is removed.

At runtime, ClickINC makes lazy enforcement for program removal to
reduce the service interruption. To remove a user
program, the program instruction dependency graph is updated and
the resource is recorded as released in ClickINC without immediate enforcement.
Meanwhile, the traffic matching rules are updated so that
the user program is not effective anymore. When a
request for adding a new program is submitted, ClickINC enforces the new updated graph as the executable on the device.
\section{Evaluation}
\label{sec:evaluation} 

We conduct experiments to display ClickINC's advantages. (1) ClickINC makes use of resources on heterogeneous devices to achieve high INC performance ($\S$\ref{subsec:benchmark});
(2) The modular programming abstraction allows more efficient INC development for users than the other solutions do, including Lyra, P4all, and P4, in terms 
of line of code and programming efficiency ($\S$\ref{subsec:obi_efficiency}); (3) the cross-device INC program allocation outperforms the current practices;
(4) ClickINC uses an efficient DP algorithm to perform program placement, achieving very short compiling time and high scalability over both the number of devices and program size;
(5) With incremental deployment, ClickINC achieves minimal impact on the network devices, traffic, and other deployed INC programs.

\subsection{Experiment Setting}
\begin{figure}[t]  
	\centering  
	\includegraphics[width=0.41\textwidth]{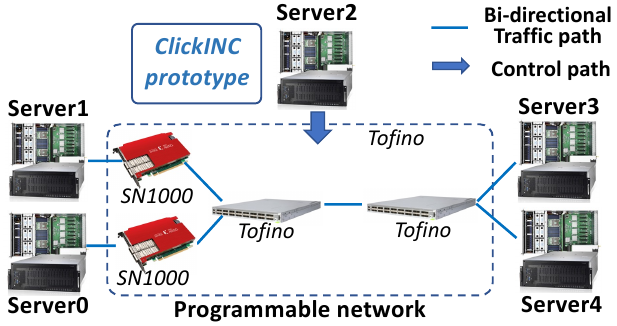} 
	\vspace{-0.3cm}
	\caption{Testbed.}
	\label{fig:testbed}  
	\vspace{-0.7cm}
\end{figure}
\parab{Implementation.}
The ClickINC framework is implemented in C++ and Python with 8,755 and 3,133 lines of code, respectively, and runs on a desktop with an Intel Core i7 4GHz CPU and 16GB RAM. It currently supports Tofino, Tofino2, TD4, Netronome smartNIC, Xilinx FPGA, covering the target DSL of P4$_{16}$, NPL, Micro-C, and Verilog HDL. 

\parab{Emulator.}
We construct a software emulation platform for evaluating ClickINC on large networks with heterogeneous devices. A server equipped with the switch SDE~\cite{p4sde} for Tofino series ASIC and BCM simulator for TD4 can emulate all the chip functions. Using virtual NIC pairs to act as switch ports, the emulator presents the same resource constraints as a real switch and can be controlled using the same API. Xilinx and Netronome also provide the software behavioral model/simulator to emulate hardware FPGA/NFP smartNIC which takes PCAP files as input and output. We set up an emulator using 4 servers with 16 Virtual Machines (VM) (4 for Tofino2, 6 for TD4, and 6 for Tofino), 4 VNetP4 behavior model instances, and 8 NFP simulator instances which are organized in the topology as shown in Fig.~\ref{fig:emulation}. The communication between VMs is bridged through the physical NIC. Communication with the VNetP4 behavior models or NFP simulators is achieved by using a script program to generate and interpret the PCAP files. 

\parab{Testbed.}
As shown in Fig.~\ref{fig:testbed}, Server2 runs the ClickINC controller and serves as the switch controller as well. Server3 and Server4 run DPDK on Mellanox ConnectX-5 dual-port 100G NIC. 
Equipped with Xilinx Alveo U280 FPGA and Netronome Agilio LX smartNIC, respectively, Server0 and Server1 generate the traffic of integer parameters. 
Two Edgecore Wedge100BF-32X switches are interconnected, and each switch further connects with the two smartNICs or the two ConnectX-5 NICs, respectively.
The link capacity is 100Gbps.

\subsection{Application Performance}
\label{subsec:benchmark}
INC programs can achieve performance gain when compiled and deployed by
ClickINC.  We control the network with (1) no programmable device, (2) only smartNICs, (3) only one Tofino switch, (4) two Tofino switches, and (5)  the
smartNIC and a Tofino switch.  We deploy the sparse gradient aggregation
program in Fig.~\ref{fig:program-sparse-agg} via ClickINC in the five network
configurations.  Fig.~\ref{fig:goodput_comp} shows the aggregation goodput and
Fig.~\ref{fig:latency_comp} is the corresponding INC processing latency.  Using
setting (1) as the baseline, ClickINC compiles the sparse gradient compression
on the smartNICs in case (2), which increases the goodput by reducing traffic
volume.  ClickINC compiles the aggregation on the switch in case (3), which
increases the goodput by in-network traffic aggregation.  The program performs
better with two switches in case (4) than one in case (3), because the packet
size can be larger in case (4), and ClickINC places the program on two
switches, each processing a part of packets.  And finally, with a combination
of two heterogeneous devices, the program achieves the highest runtime
throughput in case (5). 

\begin{figure}[tb]
    \subfigure[Goodput]{
        \label{subfig:dp-w/o}
		\includegraphics[width=0.215\textwidth]{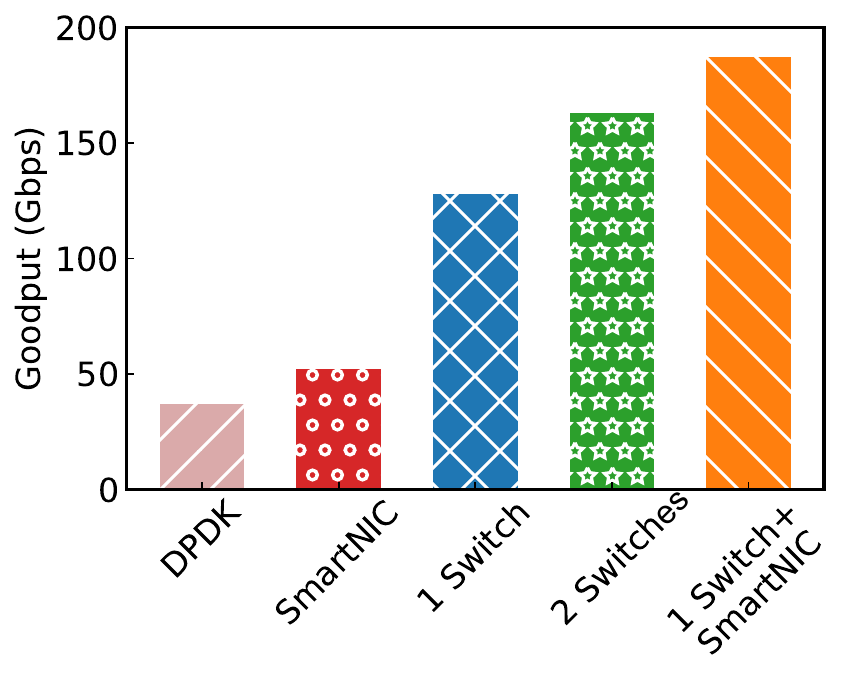} \vspace{-0.3cm}
		\label{fig:goodput_comp} }
	\hfill
	\subfigure[In-network latency]{	
	    \label{subfig:dp-block}
		\includegraphics[width=0.215\textwidth]{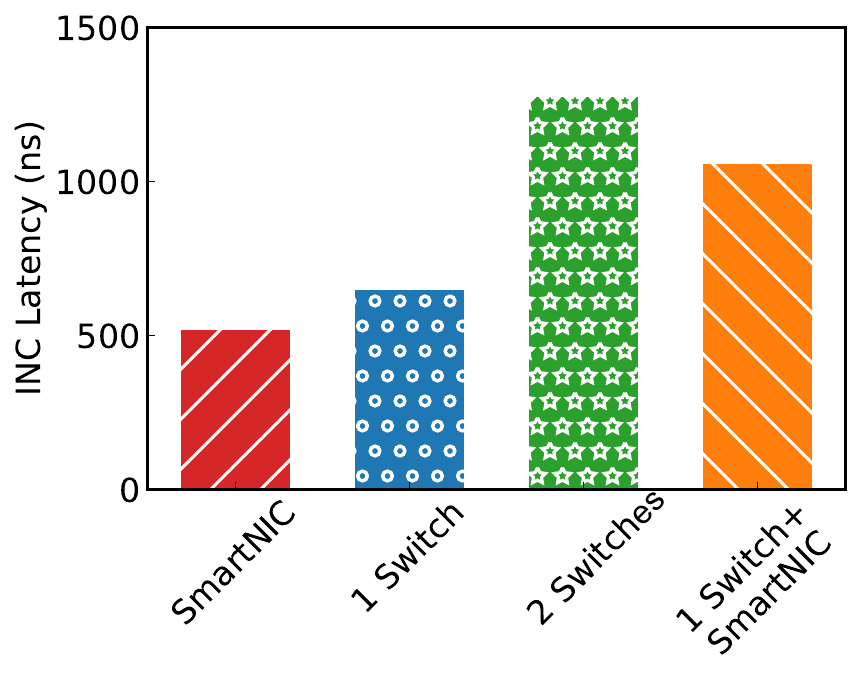} \vspace{-0.3cm}
		\label{fig:latency_comp} }
\vspace{-0.3cm}	
 \caption{Performance comparison.}	
	\label{fig:performance_comp}
	\centering
	\vspace{-0.3cm}
\end{figure}

\subsection{Program Development Workload} \label{ssec:dev-workload}
\label{subsec:obi_efficiency}

We develop three INC applications with ClickINC, Lyra, P4all, and P4$_{16}$.
The applications are (1) a KVS program using a 5K-entry cache for 128b key and
16$\times$32b value vector, and a 3$\times$1K heavy-hitter for statistics of
missed queries; (2) an MLAgg program with 5K aggregators for 24$\times$32b
integer parameter vector; (3) an SQL DISTINCT program with a 5K$\times$8
rolling cache which filters queries with 32b value.

\begin{table}[tb]
\centering
\caption{Comparison between ClickINC and other peers} \label{tab:comp_pyinc}
\vspace{-0.2cm}
\footnotesize 
\setlength{\tabcolsep}{1pt}
\begin{tabular}{|c|c|c|c|c|}
\hline
	\textbf{Language} & \textbf{LoC (KVS/} & \textbf{Modular} & \textbf{Incremental} & \textbf{Cross-Device} \\
	\textbf{}  & \textbf{MLAgg/DQAcc)} & \textbf{Programming} & \textbf{Compilation} & \textbf{Placement}\\ \hline
\textbf{ClickINC} &  16/56/13   & Y & Y & Y  \\ 
\textbf{Lyra~\cite{gao2020lyra}}  &  125/232/243   & N & N & Y  \\ 
\textbf{P4all~\cite{hoganmodular2022}}   &  202/233/138   & Y & N & N  \\ 
\textbf{P4}$_{16}$~\textbf{\cite{p4lang}}  & 571/1564/403   & N & N & N  \\ \hline
\end{tabular}
\end{table}

\parab{Program Complexity (LoC).}
Table~\ref{tab:comp_pyinc} illustrates the Lines of Code (LoC) of the three
programs in four frameworks.  ClickINC programs are 4-18, 4-12, and 28-35 times
shorter than that Lyra, P4all, and P4$_{16}$ ones, respectively. ClickINC's modular
programming reuses existing modules (outperforming Lyra), its high-level
language features (e.g., loop) are more concise, and its multi-user programming
and synthesis allows user to only write INC specific logic (outperforming
Lyra and P4all), and thus, the overall LoC is much shorter.

\parab{Developer Productivity.}

\noindent{$\bullet$ \emph{Individual Program Development.}} 
We compare the programming productivity
between using P4$_{16}$ and ClickINC on a single device. Lyra and P4all's compilers are not
publicly available when this work is done. The experimenters are experienced and originally familiar with P4 programming on Tofino, and they write the three programs.
Table~\ref{tab:single} shows the number of trials (a trial denotes a cycle of
development, compilation, test, and debug) and time spent in development.
ClickINC can reduce the development time by 6-7.2 times, and the developer
makes very few errors when developing in ClickINC (0-2 for three applications).

\begin{table}[t]
\centering
\caption{Trials and manhour in programming}\label{tab:single}
\vspace{-0.2cm}
\footnotesize
	\setlength{\tabcolsep}{2pt}
\begin{tabular}{|c|cc|cc|cc|}
\hline
\multirow{2}{*}{\textbf{Language}} & \multicolumn{2}{c|}{\textbf{KVS}} & \multicolumn{2}{c|}{\textbf{MLAgg}} & \multicolumn{2}{c|}{\textbf{MLAcc}}             \\ \cline{2-7} 
     & \multicolumn{1}{c|}{\# of trials} & time & \multicolumn{1}{c|}{\# of trials} & time      & \multicolumn{1}{c|}{\# of trials} & time      \\ \hline
\textbf{P4}$_{16}$ & \multicolumn{1}{c|}{12} & $\sim$1h & \multicolumn{1}{c|}{14} & $\sim$3h & \multicolumn{1}{c|}{6} & $\sim$30m \\
\textbf{ClickINC} & \multicolumn{1}{c|}{1} & $\sim$10m & \multicolumn{1}{c|}{2}     & $\sim$25m & \multicolumn{1}{c|}{0} & $\sim$5m  \\ \hline
\end{tabular}
\end{table} 
\begin{table}[t]
\centering
\caption{Developer Productivity of Placing Multi-user Program over Multi-devices}\label{tab:cross}
\vspace{-0.2cm}
\scriptsize
	\setlength{\tabcolsep}{2pt}
\begin{tabular}{|c|c|c|c|c|c|c|c|}
\hline
\textbf{Metrics} & \textbf{Method} & \textbf{KVS0} & \textbf{DQAcc0} & \textbf{MLAgg0}  & \textbf{DQAcc1} & \textbf{MLAgg1} & \textbf{KVS1} \\ \hline
\multirow{2}{*}{\textbf{\# of trials}}   & \textbf{P4}$_{16}$   & 2   & 16  & 25  & 31  & 24 & 13 \\ \cline{2-8} & \textbf{ClickINC} & \multicolumn{6}{c|}{1} \\ \hline
\multirow{2}{*}{\textbf{Time}} & \textbf{P4}$_{16}$ & $\sim$5m & >1h & >4h & >3h & >2h  & $\sim$1h   \\  
\cline{2-8} & \textbf{ClickINC} & \multicolumn{6}{c|}{<10s}  \\ \hline
\multirow{2}{*}{\textbf{Device}} & \textbf{P4}$_{16}$   & ToR5 & \begin{tabular}[c]{@{}c@{}}ToR0,1;\\ Agg0,1\end{tabular} & \begin{tabular}[c]{@{}c@{}}Agg0,1;\\ Agg4,5\end{tabular} & \begin{tabular}[c]{@{}c@{}}ToR1,2;\\ Agg4,5\end{tabular} & \begin{tabular}[c]{@{}c@{}}ToR2,3;\\ Agg2,3\end{tabular} & Cores          \\ \cline{2-8}  & \textbf{ClickINC}  & ToR5  & \begin{tabular}[c]{@{}c@{}}ToR0,1;\\ ToR5\end{tabular}   & \begin{tabular}[c]{@{}c@{}}Agg4,5;\\ ToR5\end{tabular}   & \begin{tabular}[c]{@{}c@{}}ToR2;\\ Agg0,1\end{tabular}   & \begin{tabular}[c]{@{}c@{}}ToR2,3;\\ Agg2,3\end{tabular} & Cores          \\ \hline
\multirow{2}{*}{\textbf{Resource}} & \textbf{P4}$_{16}$ & 1 & 2 & 2.25 & 2 & 2 & 4\\ 
\cline{2-8} & \textbf{ClickINC} & 1 & 1.71 & 1.5 & 3 & 2 & 4\\ \hline
\multirow{2}{*}{\textbf{Comm.}} & \textbf{P4}$_{16}$ & 0 & 0.75 & 0.14 & 0.63 & 0.14 & 0\\ 
\cline{2-8} & \textbf{ClickINC} & 0 & 0.33 & 0.16 & 0 & 0.14 & 0\\ \hline
\end{tabular}
\end{table}

\noindent{$\bullet$ \emph{Multi-user Program Placement and Synthesis.}} 
With the three individual programs ready, we further let two students place multiple
instances of the programs into the network, one with ClickINC and another manually. The topology is in Fig.~\ref{fig:emulation},
and all devices are assumed to be Tofino switches.
There are six INC program instances: (1) KVS0, processing traffic 
\{pod0(a), pod1(a)\} $\to$ \{pod2(b)\},
(2) DQAcc0, 
\{pod0(a), pod0(b)\}$\to$\{pod2(b)\},
(3) MLAgg0, 
\{pod0(b), pod1(b)\}$\to$\{pod2(b)\},
(4) DQAcc1, 
\{pod0(b), pod1(a)\}$\to$\{pod2(b)\},
(5) MLAgg1, 
\{pod1(a), pod1(b)\}$\to$\{pod2(b)\},
and (6) KVS1, 
\{pod0(b), pod1(b)\}$\to$\{pod2(b)\}.
Table~\ref{tab:cross} shows the final placement results, including the time
consumption and trials, and the placed devices, normalized resource
consumption, and communication overhead.

In the beginning, manually placing a program instance on multiple devices is trivial,  e.g., KVS0 on ToR5, because all devices have abundant resources and the program does not need partition. But the placement process gradually slows down as the resource usage among devices becomes unbalanced, and the placement
needs to jointly consider partition legality, resources availability,
communication overhead, and load balancing.  For example, it takes more than
one and four hours to place DQAcc0 and MLAgg0, respectively.

In contrast, ClickINC automatically finds the optimal
placement plan, and synthesizes the programs. The process is fast ($<10s$ for six instances),
and error-free.

\subsection{Effectiveness of Placement Algorithm}
\label{sec:compiler}

\begin{table}[tb]
\centering
	\caption{Placement Plan from DP and SMT algorithms}
	\label{tab:alloc}
\vspace{-0.2cm}
	\setlength{\tabcolsep}{3pt}
\footnotesize

\begin{tabular}{|c|c|cc|cc|cc|}
\hline
\multirow{2}{*}{\begin{tabular}[c]{@{}c@{}}\textbf{INC}\\ \textbf{program}\end{tabular}} & \multirow{2}{*}{\begin{tabular}[c]{@{}c@{}}\textbf{depen-}\\ \textbf{dency}\end{tabular}} & \multicolumn{2}{c|}{\textbf{stages}}     & \multicolumn{2}{c|}{\textbf{instructions}} & \multicolumn{2}{c|}{\textbf{time (s)}}       \\ \cline{3-8} 
&  & \multicolumn{1}{c|}{SMT} & DP & \multicolumn{1}{c|}{SMT}  & DP & \multicolumn{1}{c|}{SMT} & DP \\ \hline
\textbf{KVS} & 6 & \multicolumn{1}{c|}{8}    & 8 & \multicolumn{1}{c|}{42} & 42 & \multicolumn{1}{c|}{961}    &  1.306 \\ 
\textbf{MLAgg} & 14 & \multicolumn{1}{c|}{[8,6]}    & [6,8] & \multicolumn{1}{c|}{[14,11]} & [10,15] & \multicolumn{1}{c|}{559} & 0.754 \\
\textbf{DQAcc} & 6 & \multicolumn{1}{c|}{[8,8,1]} & [6,8,3] & \multicolumn{1}{c|}{[39,21,1]} & [35,16,10] & \multicolumn{1}{c|}{160} & 0.081\\ \hline
\end{tabular}
\begin{tablenotes}
        \footnotesize
		\item[1] `$[x, y, ...]$' in the stage column means that the devices in the chain use $x, y, ...$ stages, respectively; `$[x, y, ...]$' in the instructions column means that the devices in the chain are assigned $x, y, ...$ instructions, respectively.
\end{tablenotes}
\vspace{-0.4cm}
\end{table}
\parab{Optimality.} 
We compare the result of ClickINC's DP-based allocation algorithm with the
Z3~\cite{moura2008z3} SMT-based one that is used in existing
solutions~\cite{gao2020lyra}.  As the SMT solver is unable to handle a
multi-path topology in an acceptable time, we use a simple chain with four Tofino switches, each switch with 8 pipeline stages. We place the three programs
(\S\ref{ssec:dev-workload}) and measure the algorithm execution time and 
resource usage in the placement plan.  We set the same optimization goal as
Eq.~\ref{eq:op_target} for both algorithms.  The result is in
Table~\ref{tab:alloc}.

The DP algorithm has a similar effect as the Z3 one in terms of
resource consumption and the number of involved devices. But DP algorithm runs nearly one thousand times faster, thanks to the pruning technique.

Usually, a longer instruction dependency with fewer instructions indicates a
smaller enumeration space and thus a lower processing time. This explains why MLAgg has a much shorter processing time than KVS. On the other hand, KVS has many independent stateful operations (for realizing cache) per dependency level, which degrades the pre-pruning effect, and thus consumes more compiling time than DQAcc, even though it has fewer instructions than DQAcc. 

In addition, we also test the SMT algorithm without the optimization goal.  As
a result, it saves about half of the searching time as the algorithm only searches
for a feasible solution; but it incurs larger communication overhead as the
program is partitioned across all devices.

\begin{figure}[tb]
    \subfigure[DP: w/o-Block denotes no block construction]{
        \label{subfig:dp-w/o}
		\includegraphics[width=0.16\textwidth]{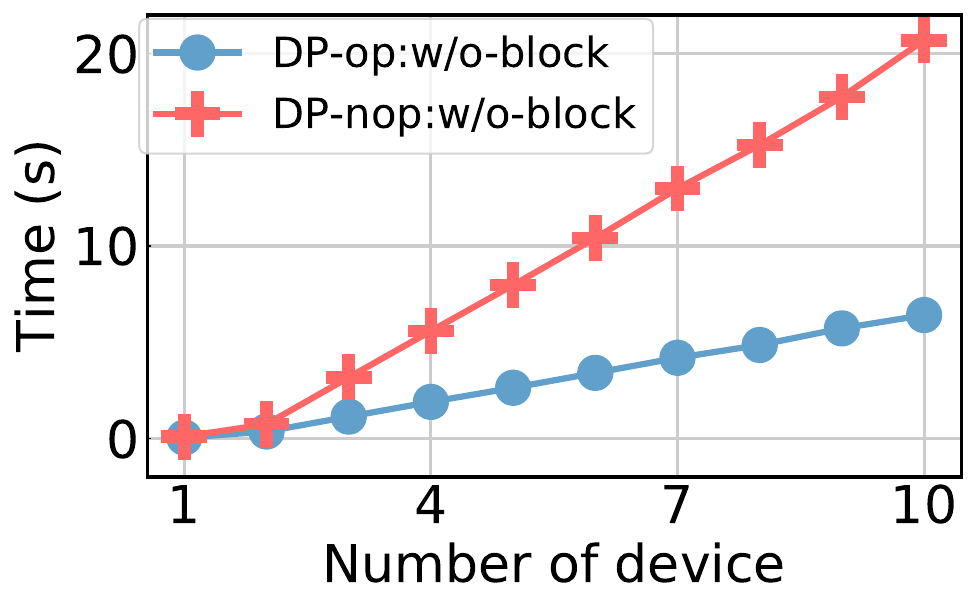} \vspace{-0.1cm}
		\label{fig:ctime_dpb_divice} }
	\hfil
	\subfigure[DP: with block construction (nop: no pruning)]{	
	    \label{subfig:dp-block}
		\includegraphics[width=0.16\textwidth]{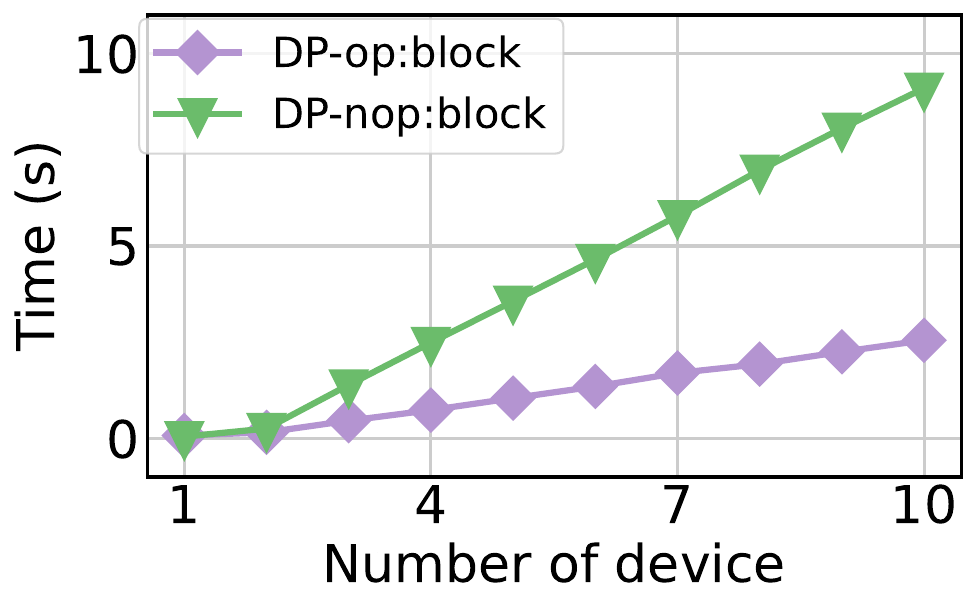} \vspace{-0.1cm}
		\label{fig:ctime_dpo_divice} }
    \hfil
	\subfigure[SMT]{
	    \label{subfig:smt}
		\includegraphics[width=0.10\textwidth]{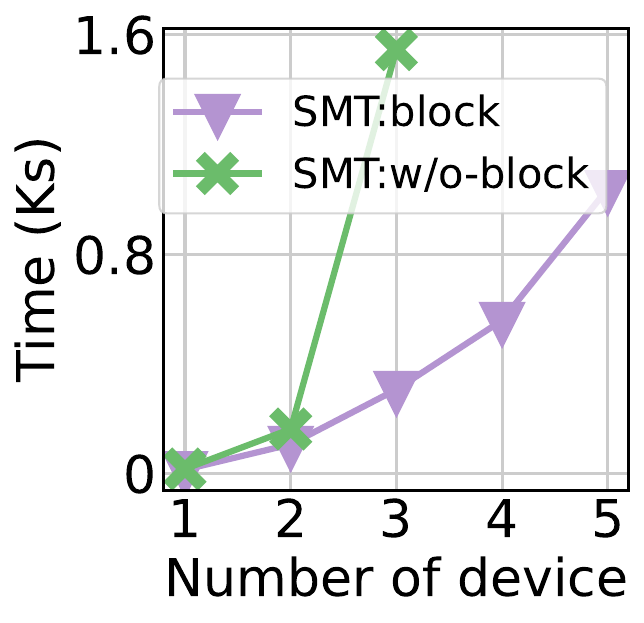} \vspace{-0.1cm}
		\label{fig:ctime_smt_device} }
	\vspace{-0.4cm}
	\caption{Compiling time on the number of devices. }	
	\label{fig:compile_time}
	\centering
	\vspace{-0.4cm}
\end{figure}

\parab{Impact of Block Construction and Pruning.}  We compile MLAgg with different 
settings of enabling/disabling block construction and pruning and measure the compilation time.
Fig.~\ref{subfig:dp-w/o} and Fig.~\ref{subfig:dp-block} show the results. The two approaches
can reduce the DP algorithm execution time by more than 50\% separately, and by more than 80\% together.
Fig.~\ref{subfig:smt} further shows that the DP algorithm has a linear processing time with the number of devices while the SMT solver has an exponential complexity.

\begin{table}[t]
\centering
\caption{Placement results with adaptive weights} \label{tab:weight_adjust}
\vspace{-0.2cm}
	\setlength{\tabcolsep}{1pt}
\scriptsize
\begin{tabular}{|c|ccccccc|}
\hline
\multirow{2}{*}{}                                                   & \multicolumn{7}{c|}{\textbf{Devices/(instructions)}} \\ \cline{2-8} & \multicolumn{1}{c|}{\textbf{MLAgg0}} & \multicolumn{1}{c|}{\textbf{KVS0}} & \multicolumn{1}{c|}{\textbf{DQAcc0}} & \multicolumn{1}{c|}{\textbf{MLAgg1}}              & \multicolumn{1}{c|}{\textbf{KVS1}} & \multicolumn{1}{c|}{\textbf{DQAcc1}}  & \textbf{MLAgg2} \\ \hline
\textbf{\begin{tabular}[c]{@{}c@{}}Fixed\\ weight\end{tabular}}    & \multicolumn{1}{c|}{\begin{tabular}[c]{@{}c@{}}ToR0:\\ ToR5\\ /(6:60)\end{tabular}} & \multicolumn{1}{c|}{\begin{tabular}[c]{@{}c@{}}ToR0:\\ ToR5\\ /(34:47)\end{tabular}} & \multicolumn{1}{c|}{\begin{tabular}[c]{@{}c@{}}ToR0:\\ {[}Agg0,1{]}\\ /(3:25)\end{tabular}} & \multicolumn{1}{c|}{\begin{tabular}[c]{@{}c@{}}{[}Agg0,1{]}:\\ {[}Agg4,5{]}\\ /(7:59)\end{tabular}} & \multicolumn{1}{c|}{\begin{tabular}[c]{@{}c@{}}{[}Cores{]}:\\ {[}Agg4,5{]}\\ /(27:54)\end{tabular}}  & \multicolumn{1}{c|}{\begin{tabular}[c]{@{}c@{}}{[}Cores{]}\\ /(28)\end{tabular}}  & / \\ \hline
\textbf{\begin{tabular}[c]{@{}c@{}}Adapt.\\ weight\end{tabular}} & \multicolumn{1}{c|}{\begin{tabular}[c]{@{}c@{}}{[}Cores{]}\\ /(66)\end{tabular}}    & \multicolumn{1}{c|}{\begin{tabular}[c]{@{}c@{}}ToR0\\ /(81)\end{tabular}}            & \multicolumn{1}{c|}{\begin{tabular}[c]{@{}c@{}}ToR5\\ /(28)\end{tabular}}                   & \multicolumn{1}{c|}{\begin{tabular}[c]{@{}c@{}}{[}ToR0:5{]}\\ /(33:33)\end{tabular}}                & \multicolumn{1}{c|}{\begin{tabular}[c]{@{}c@{}}ToR0:{[}Agg\\ 0,1{]}:ToR5\\ /(13:49:19)\end{tabular}} & \multicolumn{1}{c|}{\begin{tabular}[c]{@{}c@{}}{[}Agg4,5{]}\\ /(28)\end{tabular}} & \begin{tabular}[c]{@{}c@{}}ToR0:{[}Agg0,1{]}\\ :{[}Cores{]}:{[}Agg4,\\ 5{]}/(10:4:20:32)\end{tabular} \\ \hline
\end{tabular}
\begin{tablenotes}
        \footnotesize
        \item[1] `[$\cdot$,$\cdot]$' indicates that instructions are duplicated on devices;
		\item[2] `:' indicates that instructions are partitioned on devices;
        \item[3] `/' represents INC plugin cannot be placed on any device.
        \end{tablenotes}
\vspace{-0.4cm}
\end{table}

\parab{Impact of Adaptive Weights.} We place six instances of the three programs 
MLAgg, KVS, and DQAcc on the path from pod0(a) to pod2(b) in
Fig.~\ref{fig:emulation}. The six instances are in the order of the second row
in Table~\ref{tab:weight_adjust}. We turn on and off the Adaptive Weight (AW) to
observe its effects.

In the beginning, all devices run only the base program with spare resources, and thus $\omega_r$ in AW is near zero, making KVS0
be placed on the four \textit{Core} switches due to the dominance of
$\omega_p$; but for the fixed weight (FW), it is divided on the \textit{ToR0} and \textit{ToR5}
switches to balance both communication overhead and resource consumption. 

As the placement proceeds, the remaining resources decrease, the $\omega_r$ in AW increases, and the resource consumption begins to dominate the placement.
MLAgg1 could have been fully placed on the \textit{Core} switches
but it is divided on \textit{ToR0} and \textit{ToR5}. In addition to the lower communication overhead, AW also has the advantage that the remaining resources are more concentrated on several devices than FW does, so it is more likely to hold a complete
INC program in one device in the future. This explains why MLAgg3 can be
deployed in the AW experiment but not in the FW one.

\subsection{Incremental Program Synthesis}
\label{sec:emulation}

\begin{table}[tb]
\centering
\caption{The impact of incremental deployment}\label{tab:emulation}
\vspace{-0.2cm}
	\setlength{\tabcolsep}{1pt}
	\footnotesize
\begin{tabular}{|c|ccc|ccc|}
\hline
\multirow{2}{*}{\textbf{Step}} & \multicolumn{3}{c|}{\textbf{Incremental deployment}} 
&\multicolumn{3}{c|}{\textbf{Monolithic deployment}} \\ \cline{2-7} 
&\multicolumn{1}{c|}{\begin{tabular}[c]{@{}c@{}}\textbf{Affected}\\\textbf{Devices}\end{tabular}} 
&\multicolumn{1}{c|}{\begin{tabular}[c]{@{}c@{}}\textbf{Affected}\\ \textbf{INC}\end{tabular}} 
&\begin{tabular}[c]{@{}c@{}}\textbf{Affected}\\ \textbf{traffic}\end{tabular} 
&\multicolumn{1}{c|}{\begin{tabular}[c]{@{}c@{}}\textbf{Affected}\\\textbf{Devices}\end{tabular}} 
&\multicolumn{1}{c|}{\begin{tabular}[c]{@{}c@{}}\textbf{Affected}\\ \textbf{INC}\end{tabular}} 
&\begin{tabular}[c]{@{}c@{}}\textbf{Affected}\\ \textbf{traffic}\end{tabular} \\ \hline
\textbf{+KVS} & \multicolumn{1}{c|}{\begin{tabular}[c]{@{}c@{}}2\end{tabular}} & \multicolumn{1}{c|}{0} & 3 pods & \multicolumn{1}{c|}{\begin{tabular}[c]{@{}c@{}}2\end{tabular}}  & \multicolumn{1}{c|}{0} & 3 pods\\ \hline
\textbf{+DQAcc} & \multicolumn{1}{c|}{\begin{tabular}[c]{@{}c@{}}2\end{tabular}} & \multicolumn{1}{c|}{0} & 1 pod & \multicolumn{1}{c|}{\begin{tabular}[c]{@{}c@{}}2\end{tabular}} & \multicolumn{1}{c|}{0} & 1 pod \\ \hline
\textbf{+MLAgg1} & \multicolumn{1}{c|}{\begin{tabular}[c]{@{}c@{}}4\end{tabular}} & \multicolumn{1}{c|}{1} & 1 pod & \multicolumn{1}{c|}{\begin{tabular}[c]{@{}c@{}}8\end{tabular}} & \multicolumn{1}{c|}{2} & 3 pods\\ \hline
\textbf{+MLAgg2} & \multicolumn{1}{c|}{\begin{tabular}[c]{@{}c@{}}2\end{tabular}} & \multicolumn{1}{c|}{1} & 1 pod & \multicolumn{1}{c|}{\begin{tabular}[c]{@{}c@{}}4\end{tabular}} & \multicolumn{1}{c|}{3} & 3 pods\\ \hline
\textbf{-MLAgg1} & \multicolumn{1}{c|}{\begin{tabular}[c]{@{}c@{}}4\end{tabular}} & \multicolumn{1}{c|}{1} & 1 pod & \multicolumn{1}{c|}{\begin{tabular}[c]{@{}c@{}}8\end{tabular}} & \multicolumn{1}{c|}{4} & 3 pods\\ \hline
\end{tabular}
\begin{tablenotes}
        \footnotesize
        \item[1] `+' or `-' mean to merge or remove an INC program.
\end{tablenotes}
\vspace{-0.4cm}
\end{table}

We configure the INC programs to make them resource intensive -- KVS with a
cache size of 100,000, MLAgg1 with 16-dimension floating-point parameters, and MLAgg1
with 16-dimensional integer parameters. 

KVS and MLAgg2 serve applications from
\texttt{pod0} (client) to \texttt{pod2(a)} (server) while DQAcc and MLAgg1 serve applications from \texttt{pod1} to \texttt{pod2(b)}. We assume there is always the background traffic from \texttt{pod0} and \texttt{pod1} to \texttt{pod2}.

We place KVS, DQAcc, MLAgg1, and MLAgg2 one by one. ClickINC performs incremental deployment (named ID), and we compare it with monolithic deployment (named MD). MD synthesizes and recompiles old and new programs each time. Table~\ref{tab:emulation} shows the placement results.

In the beginning, ID and MD behave in the same way.  KVS is placed on
\textit{Agg4,5} which have a bypassed FPGA to help host a huge cache. As
\textit{Agg4,5} are sitting on the path of traffic from \{\texttt{pod0},
\texttt{pod1}\} $\to$ \texttt{pod2}, all traffic will be interrupted during program loading on \textit{Agg4,5}. DQAcc is placed on \textit{Agg2,3}, and thus only affects traffic of \texttt{pod1} but not KVS in \texttt{pod0}. 

When MLAgg1 is deployed, ID and MD start to behave differently.  ID chooses \textit{ToR2,3} with the FPGA NIC (for floating-point calculation) and only affects traffic of \texttt{pod1} including DQAcc program; MD decomposes the synthesized program of MLAgg1 and the old DQAcc (both from \texttt{pod1} to \texttt{pod2}), which leads to instruction removal from \textit{Agg2,3} and
replacement on \textit{FNIC1,2} and \textit{ToR2,3,5} (because using
\textit{ToR5} and \textit{ToR2,3} is more resource-efficient than using
\textit{Agg2,3} and \textit{ToR2,3}), affecting all traffic and INC programs.
To place MLAgg2, ID only changes device of \textit{Agg0,1}, and affects only the traffic of \texttt{pod0} and KVS; MD needs to synthesize KVS and MLAgg2, and changes \textit{Agg0,1,4,5}, thus affecting all traffic. In summary, incremental program synthesis has a much smaller impact on traffic than that of monolithic deployment which is more likely to incur global traffic interruption.

\section{Discussion}
This section discusses ClickINC's scope and limitations.

\parab{Program isolation.}
For different INC programs on the same device, ClickINC already achieves the logical isolation and a part of security isolation, but there is still a lack of performance isolation. The logical isolation ensures the functions and resources of different INC programs on the data plane are independent and non-overlapped, which also promises a kind of security isolation, i.e., any wrong INC program will not affect others. However, ClickINC cannot defend the INC program disguised by the malicious code, which may make the system attacked. For performance isolation, ClickINC cannot ensure that the performance of different program will not interfere with each other, including throughput and latency. Therefore, these works will be the future work.   

\parab{Parameter setting.}
Towards a user-friendly programming environment, ClickINC adopts a high-level abstraction of network devices, making device hardware, resource, and topology transparent to users. However, without such knowledge, some users may be chaos in setting parameters for program especially for resource-related parameters. In this paper, ClickINC currently provides a primary parameter automatical-setting model for programs derived from the provided templates by a pre-learned empirical estimation function but cannot set parameters for user-written programs, as illustrated in Appendix~\ref{appendix:configure_template}. Next, ClickINC will design a more general model to set parameters for user-written programs according to user's performance metrics and the remained network resources.

\parab{Target users.}
ClickINC makes INC easy-to-use by application developers, isolating the roles of network operator and application developer. Although in this paper, the ClickINC framework is proposed mainly for application developers to eliminate the burden of using INC, but it is also a good programming tool for network operators. Next, ClickINC will focus on addressing developing difficulties of network operation, and integrates the programming interfaces.

\parab{Program placement.}
Although ClickINC supports multi-path program placement, it assumes the topology is fat-tree or spine-leaf, and the devices in the same EC are completely same in device type and resources, so that the topology can be simplified. 
Next, we will improve the placement algorithms on the foundation presented in this paper to support the any multi-path topology, and relax the assumption of devices. 

\parab{Supported architectures.}
Currently, ClickINC only uses FPGA as a pipeline-based device with more features compared to switches. More potential can be explored. Programmable chips with different architectures (e.g., Silicon One~\cite{silicon_one}, Spectrum~\cite{spectrum}, and Trio~\cite{yang2022using}) and target DSL (e.g., DOCA~\cite{doca}, Microcode~\cite{yang2022using}) can also be modeled and supported.
\section{Related Work}

\parab{INC Applications.} Recent INC acceleration solutions only provide a monolithic program that couples the
application functions (e.g., key-value store, application data aggregation),
the network functions (e.g., reliability, packetization), and the programming
abstraction and runtime environment of a specific platform.  
$\S$\ref{ssec:problem-statement} lists the examples of key-value store~\cite{netcache2017},
synchronous aggregation~\cite{lao2021atp,265065}, and database query~\cite{tirmazi2020cheetah,lerner2019case}. Besides, ASK~\cite{he2023generic} proposes a solution for asynchronous key-value stream aggregation.

\parab{INC Frameworks on a Single Platform.} A class of works aim to
improve the INC program development on a single platform. 
Click~\cite{monsanto2013composing} supports modular policy configuration on the control plane for traditional routers. $\mu$P4~\cite{soni2020composing} allows modular programming in data plane on PISA switches by composing reusable
libraries.  P4all~\cite{hoganmodular2022} advances modular programming by
introducing elastic parameters to be configured by the compiler based on an objective function. NetRPC~\cite{netrpc} proposes INC-enabled RPC system for simplifying INC adoption; it pre-defines several operation primitives on the switch and supports limited use cases.  These three works target on a single device.
Flightplan~\cite{sultana2021flightplan} supports the partition and distribution of a \emph{single P4 program} on heterogeneous
devices. Its program needs to be manually partitioned based on empirical decisions. 

\parab{INC Frameworks on Multiple Platforms.} The existing cross-platform frameworks target different scenarios or users, and provide different abstractions.
Lyra~\cite{gao2020lyra} is a unified language for heterogeneous devices to hide hardware differences. It helps ``network
operators'' but not as much for the application developers: (1) 
Lyra applies to programmable switches with pipeline-based ASICs;
(2) Lyra's programming abstraction couples the network operations, and multi-tenant application offloading, leading to a cumbersome development; (3) Lyra only searches for a feasible solution based on SMT solver that is time-inefficient for a large-scale network with many devices.

\section{Conclusion}
\label{sec:conclusion}
ClickINC is the first work of its kind that truly decouples the INC application development and deployment process from the network and device details. The heavy lifting of ClickINC presents a simple programming interface to users and allows users to focus on the application logic only. The clear split of duties ensures agile development and quality deployment for new applications, helping accelerate the adoption of the INC paradigm and enjoy the benefits it offers. Extensive experiments show ClickINC is superior to existing tools. 


\clearpage
\bibliographystyle{ACM-Reference-Format}
\bibliography{ClickINC}

\appendix
{\centering\section*{Appendix}}

\section{ClickINC Language}
\label{appendix:pyinc}
This section explains the details of ClickINC language.
\subsection{Templates}
\label{appendix:template}

\noindent{\textbf{KVS.}}
For KVS, it mainly contains a cache with exact-match to maintain key-value results, a counter for counting hits of each entry in cache, and a heavy hitter (count-min sketch plus bloom filter) for recording missed queries. The configurable options are: (1) the cache can be realized as using stateful array or stateless matching table, which is decided by application requirements (e.g., the value dimension and size); (2) the cache depth (same as counter), the number of counter-min sketch and bloom filters to compose a heavy hitter; and (3) the type of hash functions, and the triggering threshold of heavy hitter. All of these configurations are decided by compiler, according to profile provided by users or as default.

\noindent{\textbf{MLAgg.}}
MLAgg performs aggregation for distributed ML parameters from different works, and the structure contains multiple arrays working as \textit{aggregator} to preserve aggregated parameters, \textit{bitmap} to track workers that have been aggregated, a \textit{counter} to record the number of aggregated parameters, and \textit{sequence} to record the ID of  parameter for each stage ML job. The configurable options are: (1) whether convert the floating-point parameter to an integer one, which is decided by the accepted precision value in profile; (2) whether filters sparse parts of parameters according to ``is\_sparse'' in profile; (3) the depth of aggregator (same for bitmap, counter, sequence). The code of MLAgg is described in Fig.~\ref{fig:template-mlagg}.

\noindent{\textbf{DQAcc.}}
DQAcc provides the SQL DISTINCT in-network acceleration, mainly relying on a hash-based rolling cache, i.e., multiple arrays to store historical value, and a recorder to roll each value to be replaced by new value (to approximate LRU). The configurable options are: (1) the depth and width of the cache; (2) the type of hash algorithms.

\begin{figure}
\begin{lstlisting}[language=python, basicstyle=\ttfamily\footnotesize]
from Funclib import *
cache=Table(type="exact",keys=hdr.key,vals=hdr.val)
cms=Sketch(type="count-min",keys=hdr.key)
bf = Sketch(type="bloom-filter", keys=hdr.key)
if hdr.op == REQUEST:
  vals = get(cache, hdr.key)
  if vals != None:
    back(hdr={op:REPLY, vals:vals})
  else:
    count(cms, hdr.key, 1)
    if get(cms, hdr.key) > TH:
      write(bf, hdr.key, 1)
      copyto("CPU", hdr.key)
elif hdr.op == UPDATE:
  write(cache, hdr.key, hdr.vals)
  drop
\end{lstlisting}
	\caption{Example template of key-value store.}
\end{figure}
\begin{table}[t]
\centering
\caption{ClickINC supported function list} \label{tab:function}
\resizebox{0.48\textwidth}{!}{%
\begin{tabular}{c|c}
\hline
\textbf{kind} & \textbf{function and operations} \\ \hline
\textbf{Python built-in}& \begin{tabular}[c]{@{}l@{}}min(), max(), sum(), abs(), pow(), round(), range(), len(),\\ dict(), list(), +, -, *, /, \%, //, <, >, ==, !=, $\leq$, $\ge$, =, $\&$, $|$,\\ \^{}, $\sim$, <{}<, >{}>, and, or, not, in, not in. \end{tabular} \\ \hline
\textbf{ClickINC extension}& \begin{tabular}[c]{@{}l@{}}ceil(), floor(), sqrt(), randint(), slice() \end{tabular} \\ \hline
\end{tabular}
}
\end{table}

\subsection{Profiles}
\label{appendix:profiles}
A profile includes the following fields, and Fig.~\ref{fig:configuration-kvs} shows an example profile for KVS template.

\noindent{\textbf{App.}}
App is the dedicated ID corresponding to each template, i.e., ``KVS'', ``MLAgg'', ``DQAcc''.
\noindent{\textbf{Performance.}}
As also dedicated to templates, performance provides an optional interface for users to specify their performance requirements, as illustrated in Table~\ref{tab:profile}. For KVS, it supports an objective function ``max\_hit\_acc'' to allow users to specify the performance preference over cache hit ratio and counting accuracy of heavy hitter, and also it allows for specifying demand on cache size; for MLAgg, the precision of parameter aggregation (decides whether the conversion from floating-point number to integer is feasible), the number of aggregators, and whether the parameter is sparse can also be specified. 

\noindent{\textbf{Traffic distribution.}}
For both template and user-written program, traffic distribution is required to provide the upper limit of the querying frequency (packet per second) of each client, in the format of $\{$``client ID":``*pps", $\cdots \}$.

\noindent{\textbf{Packet format.}}
The packet format also should be provided in the profile, where the traditional network packet header below UDP protocol can be abbreviated as a name, e.g., ``ethernet/ipv4/udp'', but the application protocol header should be described in detail, e.g., ``key'':``bit\_128''.

\begin{figure}[tb]
\begin{lstlisting}[language=python, basicstyle=\ttfamily\footnotesize]
agg_seq_t = Array(row=1,size=Num_agg,w=width(hdr.seq))
bitmap_t = Array(row=1,size=Num_agg,  w=Num_worker)
agg_data_t = Array(row=len(hdr.vals), size=Num_agg, w=width(hdr.vals))
valid_t = Array(row=1, size=Num_agg, w=1)
hash_f = Hash(key=hdr.seq, ceil=Num_agg)
index = read(hash_f, hdr.seq)
seq = read(agg_seq_t, index)
isvalid = read(valid_t, index)
delete = 0, overflow = 0
if hdr.op == ACK:
  if isvalid and seq == hdr.seq:
    delete = 1
    forward(hdr)
else:
  if !isvalid and !hdr.overflow:
    write(agg_seq_t, index, hdr.seq)
    write(bitmap_t, index, hdr.bitmap)
    write(agg_data_t, index, hdr.data)
    write(valid_t, index, 1)
  elif seq == hdr.seq:
    bitmap = bitmap_t.read(index)
    if bitmap & hdr.bitmap == 0:
      vals = agg_data_t.read(key=index)
      new_vals = vals + hdr.data
      for i in range(vals):
        if new_vals[i] < 0: 
          overflow = 1
          delete = 1
      new_bit = bitmap|hdr.bitmap
      if overflow:
        mirror(hdr={'bitmap':bitmap, 'data':vals,'overflow':1})
        forward(hdr)
      elif new_bit = 2^Num_worker-1:
        back(hdr={'op':REQ,'bitmap':new_bit,'data':new_vals})
        delete =1
      else:
        write(agg_data_t,index,new_vals)
        write(bitmap_t,index,new_bit)
        drop()
    else:
      forward(hdr)
if delete:
  del(agg_seq_t, index)
  del(bitmap_t, index)
  del(agg_data_t, index)
  del(valid_t, index)
\end{lstlisting}
	\caption{Example template of MLAgg}
	\label{fig:template-mlagg}
\end{figure}

\begin{table}[t]
\centering
\caption{Basic functional unit list for IR} \label{tab:atomic}
\resizebox{0.48\textwidth}{!}{
\begin{tabular}{c|c|c}
\hline
\textbf{Operation} & \textbf{Explanation}& \textbf{Supported devices}\\ \hline
$\bm{\_ram}$& 1D-memory accessed by index & All \\ \hline
$\bm{\_cam}$& content-addressable memory &  FPGA, NFP \\ \hline
$\bm{\_tcam}$& ternary-content-addressable memory &  FPGA, NFP \\ \hline
$\bm{\_emt}$ & stateless exact-match table  & All \\ \hline
$\bm{\_semt}$ & stateful exact-match table & FPGA, NFP \\ \hline
$\bm{\_tmt}$ & stateless ternary-match table & All \\ \hline
$\bm{\_stmt}$ & stateful ternary-match table & FPGA, NFP \\ \hline
$\bm{\_lpmt}$ & longest-prefix-match table& All \\ \hline
$\bm{\_randint}$ & achieve an integer random value &  All \\ \hline
$\bm{\_crc}$ &  CRC series hashing calculation &  All \\ \hline
$\bm{\_identity}$ & identity-map hashing & Tofino series \\ \hline
$\bm{\_aes}$ & AES series en(de)-crypto calculation & FPGA \\ \hline
$\bm{\_ecs}$ & ECS series en(de)-crypto calculation  & NFP \\ \hline
$\bm{\_checksum}$ & csum16 calculation  & All \\ \hline
$\bm{\_mirror}$ & mirroring a packet  & All \\ \hline
$\bm{\_multicast}$ & multicasting packet  & Tofino series, TD4 \\ \hline
\end{tabular}
}
\end{table}
\begin{table}[t]
\centering
\caption{Device capability abstraction} \label{tab:prmitive_block}
\resizebox{0.45\textwidth}{!}{%
\begin{tabular}{c|c}
\hline
\multicolumn{2}{c}{\textbf{Classify of instructions}} \\ \hline
$\bm{\mathcal{B}}_\textbf{\emph{IN}}$& \begin{tabular}[c]{@{}l@{}}Integer addition, subtraction; bit, logical operation; slicing.\end{tabular} \\ \hline
$\bm{\mathcal{B}}_\textbf{\emph{IC}}$ & Integer multiplication, division, modulus. \\ \hline
$\bm{\mathcal{B}}_\textbf{\emph{CA}}$ & \begin{tabular}[c]{@{}l@{}} Floating-point arithmetic and other complex arithmetic.\end{tabular} \\ \hline
$\bm{\mathcal{B}}_\textbf{\emph{SO}}$ & Stateful array operations. \\ \hline
$\bm{\mathcal{B}}_\textbf{\emph{EM}}$ & Exact-match table. \\ \hline
$\bm{\mathcal{B}}_\textbf{\emph{SEM}}$ & Stateful exact-match table. \\ \hline
$\bm{\mathcal{B}}_\textbf{\emph{NEM}}$ & (Ternary, LPM)-match table. \\ \hline
$\bm{\mathcal{B}}_\textbf{\emph{SNEM}}$ & Stateful (ternary, LPM)-matching table. \\ \hline
$\bm{\mathcal{B}}_\textbf{\emph{DM}}$ & Direct-match table. \\ \hline
$\bm{\mathcal{B}}_\textbf{\emph{BPF}}$ & Drop, send, copyTo. \\ \hline
$\bm{\mathcal{B}}_\textbf{\emph{APF}}$ & Mirror, multicast. \\ \hline
$\bm{\mathcal{B}}_\textbf{\emph{AF}}$ & Hash functions (CRC8, CRC16, ...), checksum. \\ \hline
$\bm{\mathcal{B}}_\textbf{\emph{CF}}$ & (En, De)-crypto. \\ \hline
\end{tabular}
}
\end{table}
\subsection{Configuring a Template}
\label{appendix:configure_template}
The modules and templates usually need to allocate resources on devices, e.g., switch register memory. From the applications' perspective, the resource allocation influences the end users' performance. During the user program development, ClickINC has no idea about the runtime resource requirement; the users may not have the knowledge about how to allocate switch resources and their influences on the performance.

For certain applications, ClickINC can derive the resource requirements
directly from the performance metric; for example, an MLAgg switch memory should
equal to its bandwidth-delay product~\cite{265065}.  There are
applications without the mathematical models to derive the resource
requirements from the performance metric. ClickINC provides a learning-based approach:
it maintains historical records of given parameter $\textbf{x}$ and the performance $\textbf{y}$,
and learns the performance estimation function $\textbf{y}=f(\textbf{x})$ (e.g., $f(\cdot)$ could be a neural network and the learning method can be SGD). 
When a user submits a configuration with application performance metric, ClickINC
searches for the parameter $\textbf{x}$ with minimum resource allocation that satisfies the performance requirements $\textbf{y}$.
\begin{equation}
\label{eq:parameter}
	\min_{\textbf{y}=f(\textbf{x})} g(\textbf{x}, \textbf{y}), \ \ s.t.
	\bigwedge_{i\in [1,k]} h_i(\textbf{x},\textbf{y})\leq 0,
\end{equation}
where $k$ is the number of performance metric constraints, $g(\textbf{x}, \textbf{y})$ means
the resource consumption, and $h_i(\textbf{x}, \textbf{y})$ means the $i$-th dimension of performance metric is satisfied. The optimization problem can be solved using gradient descent.

\subsection{Intermediate Representation}
\label{appendix:IR}
The syntax of IR is described in Fig.~\ref{fig:IR-instruciton-set}, where the \emph{instance} and \emph{action} are the basic functional units listed in Table~\ref{tab:atomic}. These units can be further utilized by network operator to write a new \emph{object} and \emph{primitive} module in Fig.~\ref{fig:language-grammar} to update the library, which are provided to developers for programming with frontend language. Although the devices in the same architecture share some common constraints, they exhibit their exclusive features as well due to particular resource requirements, e.g., Trident4 supports the en(de)-cryption while Tofino does not. Therefore, to map instructions to the correct devices, we abstract the device capability in form of atomic operations (e.g., CRC calculation) that are listed in Table~\ref{tab:atomic}, and classify them into different types as shown in Table~\ref{tab:prmitive_block}, which helps to rule out impossible mappings during allocation. 
\begin{figure}[t]
\begin{lstlisting}[ 
	% float=t,
	basicstyle=\ttfamily\footnotesize,
	% label=Code5,
	backgroundcolor=\color{white},
	keywordstyle=\rmfamily\bfseries, 
	morekeywords={instance, action}, 
	mathescape=true, 
	% columns=flexible
	]
Prog ::== Declare | Operation
Declare ::== header | parse | data | instance
header ::== h_type string {hBody}
hBody ::== struct {hFields}
hFields ::== type<length> string
type ::==  int | float | bit | bool
length ::== 1,2,...,1024
parse ::== cond? extract(hBody)
data ::== type string
instance ::== emt | semt| tmt | stmt | lpmt | cam | tcam | ram
Operation ::== cond? statement | statement
statement ::== data = operand | operand
operand ::== data calc | instance action
action ::== write | get | drop | mirror | multicast | randint | crc
        | aes | ecs| calc
calc ::== + | - | * | / | % | bit operation | >>const | <<const
condition ::== state | state&&state | state||state
state ::== data compare
compare ::== > | >= | == | <= | <
\end{lstlisting}
	\caption{IR instruction syntax}
	\label{fig:IR-instruciton-set}
\end{figure}

\begin{table}[t]
\centering
\caption{INC profile} \label{tab:profile}
\resizebox{0.48\textwidth}{!}{
\begin{tabular}{c|c|c|c|c|c}
\hline
\textbf{Template} & \textbf{KVS}& \textbf{MLAgg} & \textbf{DISAcc} & \textbf{OPSketch} & \textbf{DDoSAD}\\ \hline
\begin{tabular}[c]{@{}l@{}}$\bm{Perfo}$-\\$\bm{rmance}$\end{tabular} & \begin{tabular}[c]{@{}l@{}}``max\_hit\_acc''\\: [0.7, 0.3], \\``depth''\\:  >= 1000 \end{tabular} & \begin{tabular}[c]{@{}l@{}}``precision\_dec''\\: 3\\ ``is\_sparse'': 0,\\ ``depth'': >= 500 \end{tabular} & \begin{tabular}[c]{@{}l@{}}``c\_depth''\\: >= 1500 \\ ``c\_len''\\: >=8 \end{tabular} & \begin{tabular}[c]{@{}l@{}}``c\_depth''\\: >= 5 \\ ``c\_len''\\: >=800 \end{tabular} & \begin{tabular}[c]{@{}l@{}}``c\_depth''\\: >= 10 \\ ``c\_len''\\: >=2000 \end{tabular} \\ \hline
\end{tabular}
}
\end{table}

\section{Theories on placement}
\subsection{Analysis of Program Partitioning} \label{sec:partition_theory} 
To ensure the correctness of process on program
partitioning and instruction block construction, we provide the following
theory.  First, we define the \emph{partitioning legality} as:
\begin{definition}
Given the partitions of IR program $\mathcal{P}$, $\forall p_1, p_2 \in \mathcal{P}$, there is no bidirectional traffic flow, i.e., $p_1 \nLeftrightarrow p_2$. 
\end{definition}
The partitioning legality ensures that any two partitions can be separately placed on different devices. 

\noindent{\textbf{Program partitioning.}} 
The data in program is two kinds: (1) stateless data which is new for each round program execution and the data change will not affect the next packet, e.g., an intermediate variable; (2) stateful data, which is same for all packets and the data change affects the next packet, e.g., a cache table. To ensure the data consistency and correctness, stateful data cannot be duplicated. Thus, the instructions with operations on the same stateful data (we call them \emph{state-sharing} instructions) cannot be partitioned on different devices, i.e.,:
\begin{lemma}
    $\forall$ two instructions $p_1$ and $p_2$, if they are state-sharing, the partitioning legality is unsatisfied.
\end{lemma}
\begin{proof}
    Assume $p_1$, $p_2$ are placed on upstream device and downstream device respectively, if the stateful data is located on upstream device, then after $p_1$ is executed, the traffic flow to downstream device to execute $p_2$ which however needs to return to upstream device for accessing the stateful data, causing bidirectional traffic flow and violates partitioning legality; if the stateful data is located on downstream device, then traffic will flow to downstream device to access stateful data and return upstream device to complete $p_1$, and also violates partitioning legality.
\end{proof}
Thus, we need to group all state-sharing instructions together as an inseparable partition. Following this, we construct a directed graph for IR program as $G$, where we the vertex is inseparable state-sharing instruction partition or each other nornal instruction, and the edge to describe instruction dependency.

As long as two instruction has direct dependency (i.e., the next instruction uses the value generated by the previous instruction), we use an directed edge to connect them from previous instruction to the next one. For example, $p_1 \rightarrow p_2$ indicates the instruction $p_2$ directly depends on $p_1$. Obviously, instructions with direct dependency represents there exists data flow (the left value of $p_1$ flows to one of the right values of $p_2$), i.e., $p_1 \rightarrow p_2$ can infer that $p_1 \Rightarrow p_2$, based on which we have:
\begin{lemma}
\label{lemma_p0}
    The instruction $p_2$ depends on $p_1$ is equaling to $p_1 \Rightarrow p_2$.
\end{lemma}
\begin{proof}
    We first prove that $p_2$ depending on $p_1$ can infer $p_1 \Rightarrow p_2$.
    If $p_1$ has direct dependency with $p_2$, i.e., $p_1 \rightarrow p_2$, obviously it equals to $p_1 \Rightarrow p_2$; if $p_2$ indirectly depends on $p_1$, we assume there exists an instruction $p_a$ that has direct dependency with $p_1$ and $p_2$, i.e., $p_1 \rightarrow p_a$ and $p_a \rightarrow p_2$. Then, we have $p_1 \Rightarrow p_a \Rightarrow p_2$, thus $p_1 \Rightarrow p_2$ and the statement is proved. Last, we should prove that $p_1 \Rightarrow p_2$ can infer that $p_2$ depends on $p_1$. If the left value of $p_1$ flows to $p_2$, obviously $p_1 \rightarrow p_2$; otherwise, we similarly assume an instruction $p_a$, and the left value of $p_1$ flows to $p_a$ and $p_a$'s left value flows to $p_2$, and thus we have $p_1\rightarrow p_a$ and $p_a\rightarrow p_2$, i.e., $p_2$ indirectly depends on $p_1$, the Lemma~\ref{lemma_p0} is proved.
\end{proof}
Then, we have the following lemma:
\begin{lemma}
\label{lemma_p1}
    Directed acyclic IR dependency graph satisfies the partitioning legality.
\end{lemma}
\begin{proof}
    Assume that the acyclic dependency graph violates the partitioning legality, i.e., $\exists$ instructions $p_1$, $p_2$, $p_1\Rightarrow p_2$ and $p_2 \Leftarrow p_1$, thus we have $p_2$ depends on $p_1$ and $p_1$ depends on $p_2$ respectively according to Lemma~\ref{lemma_p0}. It means that $p_1$ and $p_2$ are cyclic in dependency, which is impossible for Directed acyclic graph (DAG). Therefore, the assumption is wrong and Lemma~\ref{lemma_p1} is proved.
\end{proof}

According to the above theory, we need to group the cyclic instruction on dependency graph as a hybrid vertex, so that becoming a IR DAG and partition legality can be always satisfied.

\noindent{\textbf{Instruction block.}}
The instruction block construction process should also maintain the partition legality. In detail, given the IR DAG $G=(V,E)$, we define the predecessor set for each vertex $v\in
V$ as $\mathcal{P}(v)=\{x\in V|<x,v>\in E\}$. We apply the Kahn’s
algorithm~\cite{kahn1962topological}, a variant of Topological sorting on $G$
to generate a series of the Kahn’s partitions
$\mathcal{K}=\{K_i\}_{i=1}^{N_K}$, where $V=\bigcup_{i=1}^{N_K} \{v|v\in K_i\}$
and $K_i\cap K_j=\emptyset (i\neq j)$. According to the Kahn’s algorithm, for
$\forall v \in K_i, i\in\{2,3,\cdots,N_K\}$, $\mathcal{P}(v) \subset
\bigcup_{l=1}^{i-1}K_l$ holds. That is, any predecessor vertex of a partition
$K$ must belong to a partition before $K$, which leads to the following lemmas.

\begin{lemma} \label{lemma0}
Given the Kahn’s partitions $\mathcal{K}=\{K_i\}_{i=1}^{N_K}$ for the DAG
	$G(V,E)$, $\forall v_m \in K_i, v_n\in K_j$, if $i>j$, then $v_m
	\nRightarrow v_n$, where $\nRightarrow$ means a node cannot reach another
	node on the graph.
\end{lemma}
\begin{proof}
Assume $\exists$ $v_m \in K_i, v_n\in K_j, i>j$ to make $v_m \Rightarrow v_n$
	hold. Then $v_m\in K_i$ is a predecessor vertex of $v_n\in K_j$ (i.e.,
	$v_m\in \mathcal{P}(v_n)$), which means it is impossible for $K_i\subset
	\bigcup_{l=1}^{i-1}K_l$. Therefore, the assumption is wrong and
	Lemma~\ref{lemma0} is proved.
\end{proof}

\begin{lemma} \label{lemma1}
Given Kahn’s partitions $\mathcal{K}=\{K_i\}_{i=1}^{N_K}$ for the DAG $G(V,E)$,
	$\forall v_m, v_n\in K_i$, if $m\neq n$, then $v_m \nRightarrow v_n$.
\end{lemma}
\begin{proof}
Assume that $\exists v_m, v_n\in K_i, i\neq j$ makes $v_m \Rightarrow v_n$.
	Then $v_m\in K_i$ is the predecessor vertex of $v_n\in K_i$ (i.e., $v_m\in
	\bigcup_{l=1}^{i-1}K_l$), which contradicts with the assumption of $v_m\in
	K_i$. Therefore, Lemma~\ref{lemma1} is proved.
\end{proof}
The following theorem is derived from the lemmas:
\begin{theorem} \label{theorem}
Given Kahn’s partitions $\mathcal{K}=\{K_i\}_{i=1}^{N_K}$ for the DAG $G(V,E)$,
	$\forall v_m\in K_{i-1}, v_n\in K_i$, if $<v_m, v_n>\in E$, then no $v_l
	\in V (l\neq m,n)$ can make $v_m \Rightarrow v_l$ and $v_l \Rightarrow
	v_n$.
\end{theorem}
\begin{proof}
Assume $\exists v_l \in K_j (l\neq m,n; j\in[1,N_K])$ that makes $v_m
	\Rightarrow v_l$ and $v_l \Rightarrow v_n$ hold. We know $j \neq i-1$ and
	$j\neq i$ from Lemma~\ref{lemma1}. Then if $j<i-1$, Lemma~\ref{lemma0}
	tells us that $v_m \in K_{i-1} \nRightarrow v_l \in K_j$, which violates
	the assumption. Similarly, if $j>i$, $v_l \in K_j \nRightarrow v_n \in K_i$
	also contradicts with the assumption. Hence, $j$ does not exist and
	Theorem~\ref{theorem} is proved.
\end{proof}
\subsection{Analysis of Device Equality}\label{appendix:proof_reduced}
Starting from the initial device status that devices at a layer in the same pod (we call them peer devices subsequently) are exactly equal in resources, we prove that these peer devices can maintain equality under our allocation algorithm.

\noindent{\textbf{Spine-leaf topology.}}
Each leaf is connected with all the same spine switches, and any path is the leaf-spine-leaf structure sharing the common spines. Thus, it's straightforward that all spines should be allocated with the same part of an INC program and regarded as the same device.

\noindent{\textbf{Full-clos Fat-tree topology.}}
For a full-clos fat-tree topology, each switch in a pod is fully connected with each of upper-layer switches which should have a higher throughput capacity, as illustrated in Fig.~\ref{fig:topo_1}. In this case, the core switches are also fully shared by all Agg switches, which is similar to spine-leaf topology and thus can also be reduced as the same device, as illustrated in the right sub-figure of Fig.~\ref{fig:topo_1}. 

Then, we should infer the equality of Agg switches in a pod. First, we denote the INC program as instruction set $[0, n]$ which should be allocated along path \texttt{pod0}-\texttt{pod1}, and we assume program placed on core switches are $[i, j],0\leq i \leq j \leq n$. Then $[0,i)$ should be placed on switches in \texttt{pod0}, and $(j,n]$ needs to be placed on switches in \texttt{pod1}. Supposing the ToR0 switch in \texttt{pod0} is allocated with instructions $[0,p]$, as ToR0 connects with all Agg switches in \texttt{pod0}, these Agg switches must be placed the same instructions $(p,n]$, making other ToR switches e allocated with $[0,p]$ correspondingly. Thus, the equality of switches at the same layer in a pod is proved. 

\noindent{\textbf{Device-equal Fat-tree topology.}}
As illustrated in Fig.~\ref{fig:topo_2}, this topology targets that device of each layer has the same throughput capacity. A $k$-fat-tree has $k$ pods and $(\frac{k}{2})^2$ core switches, and each layer in a pod has $\frac{k}{2}$ switches. In this topology, each Agg switch in a pod fully connects with the $\frac{k}{2}$ core switches, which means these core switches are shared by the current pod and can be reduced as a device. Supposing the traffic is from \texttt{pod0} to \texttt{pod1}, then we can derive the topology as the right sub-figure shows in Fig.~\ref{fig:topo_2}. 

Thereafter, we need to prove the equality of Agg devices and ToR devices in a pod. First, we still assume a instruction set $[0, n]$ to be placed along path \texttt{pod0}-\texttt{pod1}. For switches in \texttt{pod0}, we suppose the placement is $[0, p0]$ on ToR0, $[0, p1]$ on ToR1, $[q0, k0]$ for Agg0, and $[q1, k1]$ for Agg1. As ToR0 and ToR1 are both fully connected with Agg0 and Agg1, we have $p0=p1=q0=q1$, and the case for switches in \texttt{pod1} is similar. Thus ToR switches in the same pod can also be reduced as a device. Then we can derive the topology shown as the left-below sub-figure in Fig.~\ref{fig:topo_2}, i.e., multiple paths diverge from the same ToR device in \texttt{pod0} and converge at \texttt{pod1}. Fortunately, we can notice that the multiple paths are exactly the same regardless of device type, available resources. Thus, for any non-random allocation algorithm, the instructions placements on these paths are absolutely same, i.e., the allocated instructions on the Agg switches in \texttt{pod0} are exactly same, and so are core switches and Agg switches in \texttt{pod1}. That means these switches can be reduced to a single device respectively, and the topology shown as the left-below sub-figure in Fig.~\ref{fig:topo_2} converts to a chain. Thus, the equality of switches at the same layer in a pod is also proved.

\begin{figure}[t]  
	\centering  
	\includegraphics[width=0.48\textwidth]{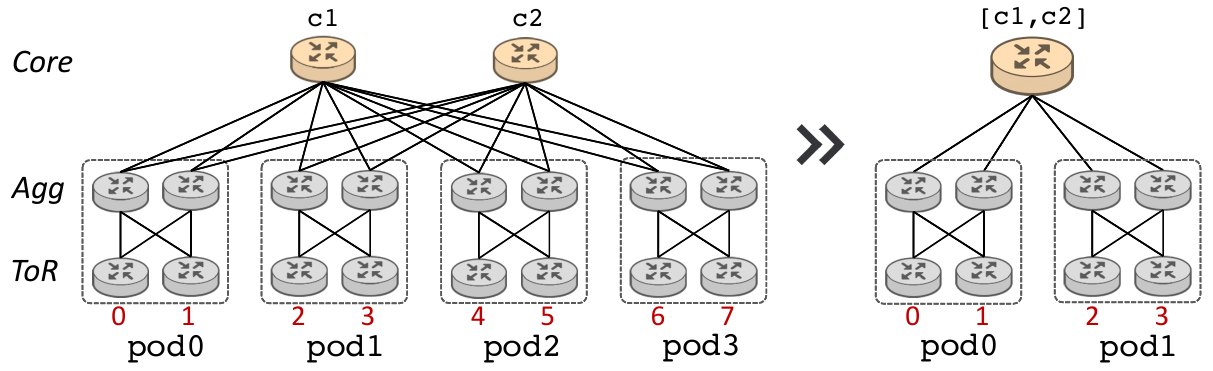} 
	\vspace{-0.3cm}
	\caption{Example of full-clos fat-tree topology.}
	\label{fig:topo_1} 
	\vspace{-0.15cm}
\end{figure}
\begin{figure}[t]  
	\centering  
	\includegraphics[width=0.48\textwidth]{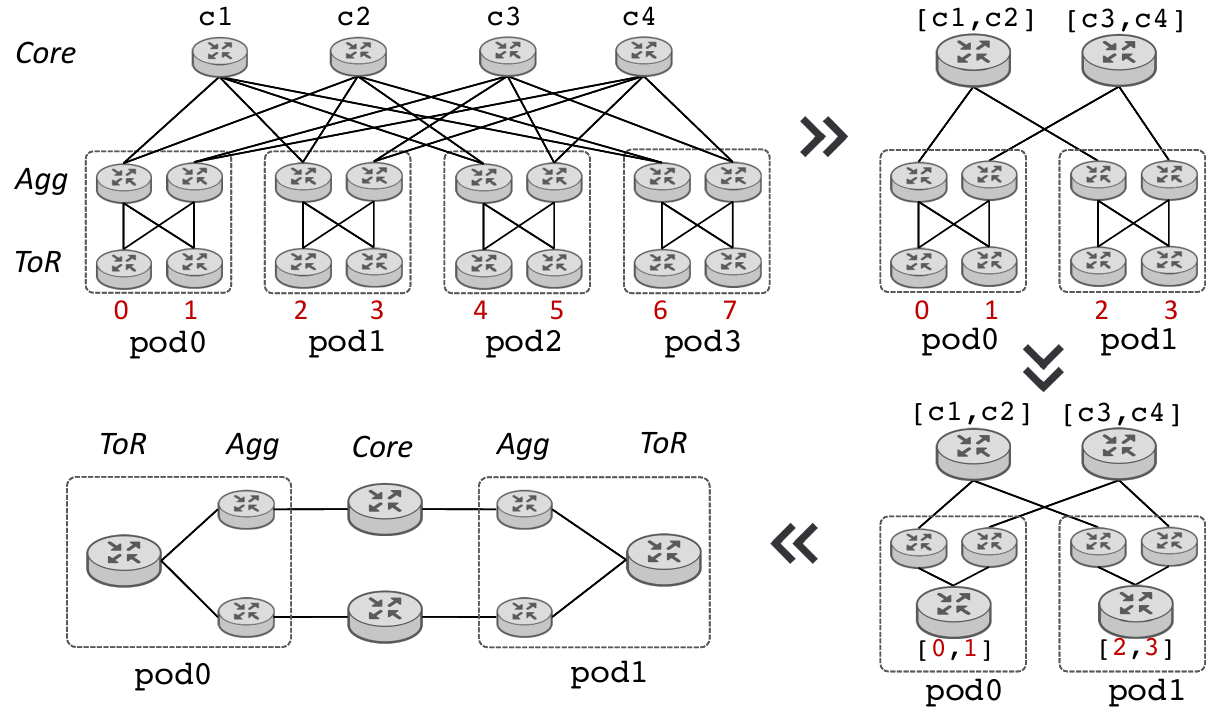} 
	\vspace{-0.3cm}
	\caption{Example of device-equal fat-tree topology.}
	\label{fig:topo_2} 
	\vspace{-0.15cm}
\end{figure}

\section{Pseudo Algorithms}
This section describes the core algorithms for program placement.
\subsection{Block construction}
\label{appendix:block_construction}
Block construction is described in Algorithm~\ref{alg:blocking}.

\subsection{Program merging}
\label{appendix:merging}
The program merging process is described in Algorithm~\ref{alg:merging}.
\begin{algorithm}[h]
	\caption{Instruction block construction}\label{alg:blocking} \footnotesize
	\KwIn {Primitives IR DAG $G=(V,E)$.}
	\KwOut {the results of blocks $G_{out}$.}
    $\mathcal{K}\leftarrow $ Kahn\_partition(G);\\
    $G_1 \leftarrow$call \texttt{intra\_partition}($G$); \\
    $G_{out} \leftarrow$call \texttt{inter\_partition}($G_1$);\\ 
    \textbf{return} $G_{out}$;\\
    \SetKwFunction{FMain}{intra\_partition}
    \SetKwFunction{FSum}{inter\_partition}
    
    \SetKwProg{Fn}{Function}{:}{}
    \Fn{\FMain{$G$}}{
    $V_1 \leftarrow \emptyset, V\leftarrow G.V; E_1 \leftarrow G.E$\;
    \For{$v_1 \in V$}{
        P $\leftarrow \{v_1\}$; remove $v_1$ from $V$\;
        \For{$v_2 \in V$}{
            \If{$v_2.type=P.type$}{
                \If{$\mathcal{K}(v_2)=\mathcal{K}(P)$}{
                \If{$\mathcal{P}(P)\cap\mathcal{P}(v_2)\neq \emptyset$}{
                    add $v_2$ to P; remove $v_2$ from $V$\;
                    combine in-edges of $P, v_2$ in $E_1$\;
                }
                }
            }
        }
        add P to $V_1$\;
    }
    \KwRet $G_1=(V_1, E_1)$;
    }
    \SetKwProg{Fn}{Function}{:}{}
    \Fn{\FSum{$G_1$}}{
    $\mathcal{K}\leftarrow $ $Kahn\_partition(G_1)$\;
    $V_2\leftarrow G1.V, E_2\leftarrow G_1.E, V \leftarrow \emptyset$\;
        \While{$|V|<|V_2|$}{
        $V_2\leftarrow V, V \leftarrow \emptyset$\;
        \For{$i$ from 0 to $|\mathcal{K}|-1$}{
        \For{$v_1 \in \mathcal{K}$[i]}{
            P $\leftarrow \{v_1\}$, remove $v_1$ from $\mathcal{K}[i]$\;
            S$\leftarrow successor(v_1)\cap \mathcal{K}[i+1]$\;
            \For{$v_2 \in S$}{
                \If{$\mathcal{K}[i+1]$}{
                    add $v_2$ to P;\\
                    remove $v_2$ from $\mathcal{K}[i+1]$\;
                        remove $<v_2, P>$ from $E_2$\;
                }
            }
        $add P to V$\;
        }
        }
        $\mathcal{K}\leftarrow Kahn\_partition$($G_2=(V,E_2)$)\;
        }
    \KwRet $G_2=(V_2, E_2)$;
    }
\end{algorithm}

\begin{algorithm}[t]
	\caption{Program merging}\label{alg:merging} \footnotesize
	\KwIn {the parsing graph of INC program and main program $T_{inc}$, $T_{main}$; the dependency graph of INC program and main program $G_{inc}$, $G_{main}$.}
	\KwOut {the whole parser and program $T_{w}$, $G_{w}$.}
    $T_{w} \leftarrow T_{main}, G_{w} \leftarrow G_{main}$;\\
    call \texttt{Parsing\_merger}($T_{inc}$, $T_{w}$);\\ 
    call \texttt{Program\_merger}($G_{inc}$, $G_{w}$);\\
    \SetKwFunction{FMain}{Parsing\_merger}
    \SetKwProg{Fn}{Function}{:}{}
    \Fn{\FMain{$T_{inc}$, $T_{w}$}}{
    \For{s \textbf{in} $T_{inc}$.traversing}{
        $t \leftarrow T_{w}.find(s), p \leftarrow T_{w}.find(s.parent)$; \\
        \If{t $=$ None}{
            $add\_son(p, s), add\_annotation(s)$;\\
            $add\_transition(p, s), add\_annotation\_in(p)$;\\
            $add\_hdr(s.hdr), add\_annotation(s.hdr)$;\\
        }
        \lElse{$add\_annotation(t)$}
    }
    }
    \SetKwFunction{FMain}{Program\_merger}
    \SetKwProg{Fn}{Function}{:}{}
    \Fn{\FMain{$G_{inc}$, $G_{w}$}}{
        \If{d$\in$ Pipeline}{
            $C_{inc} \leftarrow chain(G_{inc}), C_{w} \leftarrow chain(G_{w})$;\\
            \For{s \textbf{in} $C_{inc}$}{
                $p \leftarrow get\_ins\_position(s, C_{w})$;\\
                $C_{w}.insert(p, s), add\_annotation\_before(s)$;\\
            }
        }
        \Else{
            $G_{whole} \leftarrow merge\_DAG(G_{inc}, G_{w})$;\\
            $L \leftarrow Topological\_sort(G_{whole})$;\\
            \For{e \textbf{in} $G_{inc}$}{
                $p \leftarrow get\_level(e, L)$;\\
                $G_{w}.insert(p, s), add\_annotation\_before(s)$;\\
            }
    }}
\end{algorithm}
\section{Device modeling}
\label{appendix:device_modeling}
The architectures of programmable network devices are mainly pipeline and run-to-complete (RTC). Some devices, e.g., Netronome smartNIC and FPGA, can implements both pipeline and RTC, and we call it hybrid device.

\noindent{\textbf{Pipeline device.}} The pipeline devices, e.g., Tofino and Trident4, have fixed number of stages and resources on it, and program instructions should be placed on these stages and satisfy the constraints of stage sequence and resources. For such pipeline-based devices, we denote the available stages for allocation on device $d$ as $S_d$, and use $a_{p, s}\in\{0,1\}$ to indicate whether $p\in{v}$ should be allocated on stage $s\in S_d$. Thus, $a_{p, s}$ should satisfy the pipeline stage sequence as:
\begin{equation}
\label{eq:pipeline_stage_constraint}
        \bigwedge_{s\in S_d; p_i, p_j \in v} (\gamma_{p_i, p_j}a_{p_i,s}a_{p_j,s} = 0)
\end{equation}
where $\gamma_{p_i, p_j}$ indicates the dependency relationship between $p_i$ and $p_j$ (1 represents that $p_j$ depends on $p_i$, -1 vice versa, and 0 means $p_i$ and $p_j$ are independent), $v$ is the instruction block to be allocated on device. Correspondingly, the resource constraint also should be satisfied:
\begin{equation}
\label{eq:pipeline_resource_constraint}
        \bigwedge_{s\in S_d; r\in R_d} \sum_{p\in v} \phi(r,p)a_{p,s} \leq \Omega(r,s)
\end{equation}
where $R_d$ is the set of resource on device $d$, $\phi(r, p)$ computes the size of needed resource $r$ of instruction $p$, and $\Omega(r,s)$ is the total size of resource $r$ on the pipeline stage $s$.

Especially, pipeline devices cannot support the instructions with cyclic dependency unless using recirculation which will significantly degrade the throughput and is not allowed in this paper.

\noindent{\textbf{RTC device.}}
RTC devices, e.g., Silicon one and Spectrum,  use general processor to run the whole program for each packet until completion, and thus can support instructions in cyclic dependency. To place the program on RTC devices, only the resource constraints should be satisfied. We use $a_p\in\{0,1\}$ to denote whether $p\in v$ should be placed on the device, and we have the resource constraint as:
\begin{equation}
\label{eq:rtc_resource_constraint}
        \bigwedge_{r\in R_d} \sum_{p\in v} \phi(r,p)a_{p} \leq \Omega(r)
\end{equation}

\noindent{\textbf{Hybrid device.}}
Hybrid devices mainly include multi-core smartNIC and FPGA. The former device typically has about a hundred RTC cores, which can be programmed to work parallel for same program, or be organized as core-based pipeline where each core behaves as a pipeline stage. For FPGA, the programmability is more flexible and can also combines RTC with pipeline. For these devices, we use $a_{p,s}\in\{0,1\}$ to represent whether $p\in{v}$ should be allocated on stage $s\in S_d$ and $f_{s}$ to denote the number of parallel programs on each stage (e.g., four cores run the same part of program in the same pipeline stage). Thus we have the pipeline stage constraint same as Eq.~\ref{eq:pipeline_stage_constraint}, and the resource constraint as:
\begin{equation}
    \bigwedge_{r\in R_d} \sum_{s\in S_d}\sum_{p\in v}\phi(r,p)f_s a_{p,s} \leq \Omega(r)
\end{equation}

\section{Chip Resource constraints}
\label{appendix:contraints}
ClickINC covers the
resource constraints of four major kinds of programmable chips: Tofino series
ASIC, Trident 4 switch ASIC, Netronome Network Processor, and Xilinx FPGA. The
constraints for other programmable chips can be modeled similarly. The
available device resources are updated after the deployment or removal of each
INC program.

\subsection{Tofino ASIC}
\label{appendix:cons_tofino}
Tofino ASIC follows the RMT
\cite{bosshart2013forwarding} architecture. Each pipeline stage has a fixed
share of resources including TCAM, SRAM, ALU, PHV, etc. It supports some extern
functions such as register, hash, and checksum.

\noindent{\textbf{Compatibility.}}
Due to chip capability, Tofino series switch cannot support some operations. Such constraint for operations arranged in block is:
\begin{equation}
    \bigvee_{v \in \mathcal{B}_{IC}+\mathcal{B}_{CA}+\mathcal{B}_{DM}+\mathcal{B}_{SEM}+\mathcal{B}_{SNEM}+\mathcal{B}_{CF}; d \in Tofino} x_{v,d} = 0
\end{equation}

\noindent{\textbf{Memory.}}
Tofino ASICs have two types of block memory, TCAM and SRAM, distributed in each stage. Ternary, LPM, and range matching are implemented with TCAM plus SRAM, and exact matching is implemented with SRAM only. 
We model the TCAM constraint as: 
\begin{equation}
    \bigwedge_{s \in S_d} \big [ \sum_{v\in \mathcal{B}_\emph{{NEM}}}\sum_{p\in v}x_{v,d}a_{p,s} \lceil \frac{h(p,s)}{h_{tcam}} \rceil \lceil \frac{w_{key}(p)}{w_{tcam}} \rceil \leq M_t \big ]
\end{equation}
where $M_t$ is the number of TCAM blocks per stage, $h_{tcam}$ and $w_{tcam}$ are the depth and width of TCAM block, respectively, and $w_{key}(p)$ is the width of TCAM matching key which should be less than 528 bits due to the TCAM crossbar limitation. 
SRAM can be used for read-only matching tables or registers supporting both read and write. The respective limitation on memory size is:
\begin{equation}
    \begin{aligned}
        &\bigwedge_{s \in S_d} \big \{ \big [ \big ( \sum_{v \in \mathcal{B}_\emph{EM}} \sum_{p \in v} x_{v,d} a_{p, s} \lceil \frac{h(p,s)}{\gamma h_{sram}}\rceil \lceil \frac{w_{key}(p)+w_{val}(p)}{w_{sram}} \rceil + \\
        & \sum_{v\in \mathcal{B}_{SO}+\mathcal{B}_{NEM}} \sum_{p \in v} x_{v,d} a_{p, s} \lceil \frac{h(p,s)}{h_{sram}} \rceil \lceil \frac{w_{val}(p)}{w_{sram}} \rceil \big ) \big ] \leq M_s \big \}
    \end{aligned}
\end{equation}
where $M_s$ is per-stage number of SRAM blocks, 
$h_{sram}$ and $w_{sram}$ are the depth and width of SRAM block, respectively, $\gamma \approx 90\%$ is the utilization ratio of exact-matching table for resolving hashing conflicts, $h(p,s)$ denotes the depth of exact-matching table in $\mathcal{B}_{EM}$ as well as the number of registers for primitive $p$ in $\mathcal{B}_{SO}$, $w_{key}(p)$ is the width of SRAM matching key which should be less than 1024 bits due to the hardware limitation, and $w_{val}(p)$ is the bit width of stored values.

\noindent{\textbf{Match-action Table.}}
To match the match-action pipeline architecture in Tofino, the match-action structures are constructed. 
The remaining primitives are synthesized as keyless-match action structures and those primitives within the same conditional statement are grouped in the same keyless-match action structure.

Besides, Tofino only allows a limited bit-width on conditional statements including the statements connected by logical-And (e.g., for \texttt{if(c1\&\&c2)}, the sum of bit length of \texttt{c1} and \texttt{c2} is limited). The statements within the limitation can be covered by the Gateway resources, but the statements exceeding the limitation can only be nested (e.g., \texttt{if(c1)\{} \texttt{if(c2)} \texttt{...\}}) which will be synthesised as ternary match-action using conditional statements as matching key. According to \cite{jose2015compiling}, the RMT architecture allows at most $N_{tab}$ match-action tables per stage. So the table constraint is:
\begin{equation}
\label{eq:constraints_tab}
    \bigwedge_{s \in S_d} \big( \sum_{M_a}\bigvee_{p \in M_a} a_{p, s}\leq N_{tab}\big)
\end{equation}
where $M_a$ is the set of match-action structures.

Match tables are implemented in memory blocks distributed in each stage. If the table size exceeds the memory capacity of a stage, it is spread into multiple stages. 
Thus, we have:
\begin{equation}
    \begin{aligned}
         \prod_{p \in V_m} \big ( \sum_{s\in S_d} a_{p,s} \big ) \geq 1, 
         \prod_{p \in V-V_m} \big ( \sum_{s\in S_d} a_{p,s}\big ) = 1
    \end{aligned}
\end{equation}
where $V_m = \mathcal{B}_{EM}+\mathcal{B}_{NEM}$. $H(p)=\sum_{s\in S_d}h(p, s)a_{p,s}$ is the number of entries for table primitive $p$, where $h(p,s)$ is the size of the table segment in stage $s$.  $h(p,s)$ also needs to be solved along with $x_{v,d}$ and $a_{p,s}$, given the table entry width $w(p)$, memory block depth $h$ and width $w$, and the number of memory blocks per stage $N_m$. 

\noindent{\textbf{Stateful and Stateless Operations.}}
In Tofino ASICs, stateful operations include counter, meter, and register, where the first two items are write-only for data plane and the last item supports both data-plane read and write. 
The register is composed of the register memory and a stateful arithmetic logical unit (SALU). The register memory is implemented using SRAM blocks, and the SALU integrates the arithmetic calculation and comparison operation. Registers are organized as groups, and each group can define a large amount of registers. Each stage allows a limited number of stateful operations denoted as $N_{so}$, and each stateful operation can only access a particular register using a unique group ID and register index. So we have:
\begin{equation}
    \bigwedge_{s \in S_d} \big [ \sum_{v \in \mathcal{B}_{SO}} \sum_{p \in v} x_{v,d}a_{p,s} \leq N_{so} \big ]
\end{equation}
Similarly, for stateless ALUs, we also have:
\begin{equation}
    \bigwedge_{s \in S_d} \big [ \sum_{v \in \mathcal{B}_{IN}} \sum_{p \in v} x_{v,d}a_{p,s}u_p \leq N_{sl} \big ]
\end{equation}
where $u_p$ is the number of needed ALUs of primitive $p$, and $N_{sl}$ is the limitation on stateless ALUs per stage.

\noindent{\textbf{Parser.}}
In Tofino ASICs, the parser is implemented as a state machine where each state corresponds to a packet header. The state machine uses TCAM to store 32-bit matching data and 8-bit parser state, and uses SRAM to store the next parser state and the header location. Supposing the number of TCAM entries is limited to $N_e$, we have the following constraint:
\begin{equation}
    \sum_{e\in E} \big [ \bigwedge_{h\in H(e)} V(h) \big ] \leq N_e
\end{equation}
where $E$ is the set of non-repeated transition entries from the base forwarding program and the new INC program, $H(e)$ represents all the defined packet headers with the transition entry $e$, and the Boolean function $V(h)$ indicates whether header $h$ is valid.

\noindent{\textbf{PHV.}}
The Packet Header Vector (PHV) carries the parsed packet headers and temporal metadata through the pipeline. 
Deducting the part used by the base forwarding program, we suppose the available byte length of PHV is $N_p$, and the number of $w$-bit containers for PHV is $n^p_w$, correspondingly. It is easy to see that $\sum_w n^p_w = N_p$. Given a field $f$ of length $l_f$ from a data-plane IR program $P$, it has several placement schemes on these containers, and we denote the set of all fields to be placed as $F$. For example, for a 24b variable, it may use a 16b container with a 8b container or three 8b containers. A strategy that is most likely to hold all variables must have a maximum utilization on containers, i.e., use the largest container first, and to place the variable in decreasing sequence. In detail, we should first use 32b container to try the best to place variables with length over 32b, where the remainder of the modulo 32b will be put back into the variable pool to be placed later. The 32b containers will continually be placed with values of length less than 32b in descending order as long as they are available, and the process is similar to the other containers. Thus, we have the following constraint:
\begin{equation}
    \begin{aligned}
         &L_1 = f_{top}(f_{mod}(F, 32), max\{0, n_{32}-\phi_{32}(F)\}) \\
         &m_{16} = n_{16}+[1-2sgn(max\{0, n_{32}-\phi_{32}(F)\})]\big(n_{32}-\phi_{32}(F)\big) \\
         &L_2 = f_{top}(f_{mod}(L_1, 16), max\{0, m_{16}-\phi_{16}(L_1)\}) \\
         & m_8 = n_{8}+[1-4sgn(m_{16})]\big(n_{32}+\frac{n_{16}}{2}-\phi_{32}(F)\big) \\
         & m'_8 = m_{8}+[1-2sgn(max\{0, m_{16}-\phi_{16}(L_1)\})]\big(m_{16}-\phi_{16}(L_1)\big) \\
         &\sum_{i=1}^{n_f}\lceil div(L_2, 8) \rceil \leq m'_8
    \end{aligned}
\end{equation}
where $f_{top}(L, n)$ represents that turns the top $n$ values in the set $L$ to zeros (i.e., they are already placed in containers), $f_{mod}(L, w)$ denotes the set of values in the set $L$ dividing $w$, $div(L, w)$ is the set that contains values in the set $L$ dividing $w$, $\phi_w(L)\sum_{i=1}^{n_f}\lfloor div(F, w)\rfloor$, $sgn(\cdot)$ is the sign function, $L_1$ represents the set of left variables to be placed after the first round placement, $m_{16}$, $m_{8}$ is the corresponding number of remained 16b and 8b containers; $L_2$ is the set of left variables to be placed after the second round placement, $m'_8$ is the number of remained 8b containers; and the last expression denotes the constraint for the last round placement as only 8b contains are left.

\noindent{\textbf{Other Constraints.}}
Hash computation can be directly called as extern function or implicitly executed when accessing some SRAM tables (e.g., for exact matching and register indexing). The resources of hash on each stage are limited by Hash Distribution Units and Hash bits. The former resource is needed by extern hash computation and stateful operation. We model the constraint as:
\begin{equation}
    \bigwedge_{s\in S_d} \big[ \sum_{v \in \mathcal{B}_{HF}+\mathcal{B}_{SO}} \sum_{p \in v} x_{v,d}a_{p,s} \leq N_{hd} \big]
\end{equation}
where $N_{hd}$ is the number of Hash Distribution Units per stage. For Hash bits, it will be used to store the hashing value, and we have:
\begin{equation}
    \bigwedge_{s\in S_d} \big[ \sum_{v \in \mathcal{B}_{AF}+\mathcal{B}_{EM}} \sum_{p \in v} x_{v,d}a_{p,s}w(a) \leq N_{hb} \big]
\end{equation}
where $N_{hb}$ is the number of Hash bits per stage.

Conditional operations rely on gateway resources for implementation, and gateway consumption depends on the total number of different predicates. For each spare stage, we use $N_{gw}$ to denote the gateway resources and use $C_P$ to denote the set of non-repeated predicate statements for primitives (nested predicate statements should be separated) appeared in data-plane IR program $P$. So we have:
\begin{equation}
\label{eq:constraints_gateway}
    \bigwedge_{s \in S_d} \big[ \sum_{f \in C_P} min(1, \sum_{p \in P(f)} a_{p, s}) \leq N_{gw} \big]
\end{equation}
where $P(f)$ represents the set of all primitives that hold the predicate statement $f$.

\subsection{Trident4 ASIC}
\label{appendix:cons_td4}
\noindent{\textbf{Compatibility.}}
For TD4, it also cannot support some operations:
\begin{equation}
    \sum_{v \in \mathcal{B}_{IC}+\mathcal{B}_{CA}+\mathcal{B}_{SEM}+\mathcal{B}_{SNEM}+\mathcal{B}_{CF}, d \in TD4} x_{v,d} = 0
\end{equation}
Trident4 (TD4)~\cite{td4} is a
pipeline-based programmable switch that can be programmed by Network
Programming Language (NPL)~\cite{Nplang}. The pipeline stages in TD4 have
unbalanced resources (e.g., some stages have TCAM but not SRAM), so it is more
difficult for program allocation than Tofino. The memory of TCAM and SRAM are
arranged in tiles which support exact-match, ternary-match, and index-match.
TD4 provides special components to support complex functions (e.g., mirror) and
stateful operations.

\noindent{\textbf{Memory tiles.}}
TD4 supports three kinds of match-action, i.e., non-exact-match, exact-match and index-match, which are carried on respective resource abstraction of tiles, e.g., exact-match tile. Each tile is composed of the minimum units called bank, similar to the memory block in Tofino. For non-exact-match, the banks in tile for match key and filed are independent, and the constraint is: 
\begin{equation}
    \begin{aligned}
    & \bigwedge_{s \in S_d} \big [ \max \big ( \sum_{v\in \mathcal{B}_\emph{{NEM}}}\sum_{p\in v}x_{v,d}a_{p,s} \lceil \frac{h(p,s)}{h_{key}} \rceil \lceil \frac{w_{key}(p)}{w_{key}} \rceil,\\
    &\sum_{v\in \mathcal{B}_\emph{{NEM}}}\sum_{p\in v}x_{v,d}a_{p,s} \lceil \frac{h(p,s)}{h_{field}} \rceil \lceil \frac{w_{val}(p)}{w_{field}} \rceil \big ) \leq \sum_{t\in T_{tcam}(s)} M(t) \big ]
    \end{aligned}
\end{equation}
where $h_{key}$ and $w_{key}$ are the depth and width of the bank for match key in ternary-match tile, $h_{field}$ and $w_{field}$ are the depth and width of the bank for matched filed, $h_{key}(p)$ and $w_{key}(p)$ are the depth and width of instruction $p$, $T_{tcam}(s)$ is the set of TCAM tiles in stage $s$ and $M(t)$ is the number of blanks for tile $t$. For exact-match and index-match, the key and filed share a blank and thus have the constraint as:
\begin{equation} \small
    \begin{aligned}
        \bigwedge_{s \in S_d} &\big [ \sum_{v \in \mathcal{B}} \sum_{p \in v} x_{v,d} a_{p, s} \lceil \frac{h(p,s)}{\gamma h_{bank}}\rceil \lceil \frac{w_{key}(p)+w_{val}(p)}{w_{bank}} \rceil
        \leq M_t(s) \big ]
    \end{aligned}
\end{equation}
where $\mathcal{B} \in \{\mathcal{B}_\emph{EM},\mathcal{B}_\emph{DM}\}$, $h_{bank}$ and $w_{bank}$ are the depth and width of the bank for the tile, and $M_t(s)$ is the total number of banks for related tiles in stage $s$.

\noindent{\textbf{Match-action table.}}
TD4 only adopts match-action structure for data searching, while all the other operations are abstracted and conducted in other components, e.g., flexible switch logic (FSL). For match-action table in TD4, it only allows to assign the matched field to a particular bus where all other logic (e.g., arithmetic operation) should be carried outside. The number of simultaneous match-action tables is limited by the count of banks in tiles as:
\begin{equation}
    \bigwedge_{s \in S_d, \mathcal{B} \in \{\mathcal{B}_{EM}, \mathcal{B}_{NEM}, \mathcal{B}_{DM}\}} \big [\sum_{v \in \mathcal{B}} \sum_{p \in v} x_{v,d}a_{p,s} \leq M_t(s) \big ]
\end{equation}

\noindent{\textbf{Stateful and stateless operations.}}
The state operations in TD4 are more powerful than Tofino but rely on flex state component that is only available in few stages, and the constraint is similar:
\begin{equation}
    \bigwedge_{s \in S_d} \big [ \sum_{v \in \mathcal{B}_{SO}} \sum_{p \in v} x_{v,d}a_{p,s} \leq N_{so}(s) \big ]
\end{equation}

The stateless operations are carried on the flexible switch logic which contains control logic and data logic, in a unit of \textit{floors}. Control logic floor is used for Boolean operation (e.g., comparison) and data logic is used for simple arithmetic (e.g., sub), where the usage of floors is correlated with the length of operands. For control logic floors, we have the constraint:
\begin{equation}
    \bigwedge_{s \in S_d} \big [ \sum_{f \in \mathcal{C}_{P}} \sum_{p\in P(f)} a_{p,s} \lceil \frac{w_v}{8} \rceil n_c \leq N_{fsl}^c(s) \big ]
\end{equation}
where $C_p$ is the set of Boolean operations, $\lceil \frac{w_v}{8} \rceil$ is the operand width in byte, $n_c$ is the number of processing unit of a control logic floor and $N_{fsl}^c(s)$ is the total number of control logic floors in the stage $s$. Similarly, we have the constraint for data logic:
\begin{equation}
    \bigwedge_{s \in S_d} \big [ \sum_{v \in \mathcal{B}_{IN}} \sum_{p \in v} x_{v,d} a_{p,s} \lceil \frac{w_v}{8} \rceil n_d \leq N_{fsl}^d \big ]
\end{equation}
where $n_d$ is the number of processing units of a data logic floor and $N_{fsl}^d(s)$ is the total number of data logic floors in the stage $s$.

\noindent{\textbf{Special functions.}}
TD4 embeds some special functions for complex operations (e.g., mirror) in some particular stages, thus we have the constraint:
\begin{equation}
    \bigwedge_{s \in S_d} \big [\sum_{v \in \mathcal{B}_{AF}+\mathcal{B}_{APF}} \sum_{p \in v}x_{v,d} a_{p,s} \leq N_{sf}(s) \big ]
\end{equation}

\noindent{\textbf{Parser.}} 
Different from RMT architecture whose parser is entirely in the front of pipeline, TD4 has several parser components in pipeline and match-action table is allowed between different parsers. For each parser, we have the constraint as:
\begin{equation}
    \bigwedge_{s \in S_{parse}} \big( \sum_{v \in \mathcal{B}_p} \sum_{p \in v} x_{v,d}s_{a,s} \leq N_{parse} \big )
\end{equation}
where $S_{parse}$ is the set of stages with parser component, $\mathcal{B}_p$ is the set of parsing primitives and $N_{parse}$ is the upper allowed number of parsing state transitions.

\noindent{\textbf{Buses.}}
Similar to PHV in Tofino, TD4 uses buses to carry packet data and metadata in pipeline. There are four kinds of buses: 1) field bus, which carries the fields from/to the packet; 2) object bus, which carries internally generated metadata; 3) command bus, which carries control information between components in the pipeline; 4) auxiliary bus, which carries commands and instructions that are intended for a specific function, e.g., drop. For these buses, there is upper width constraint:
\begin{equation}
    \bigwedge_{k \in \mathcal{K}_{bus}} \big[\sum_{i\in F(k)} \lceil\frac{w(i)}{w_k}\rceil \leq \frac{N_{bus}(k)}{w_k}\big ]
\end{equation}
where $\mathcal{K}_{bus}$ denotes four kind of buses, $F(k)$ represents the set of data relevant to the bus $k$, $w_k$ is the bus granularity, and $N_{bus}(k)$ is the total width of bus $k$.

\subsection{Netronome Network Flow Processor}
\label{appendix:cons_nfp}
Netronome's Agilio
smartNICs~\cite{agilio} are based on run-to-complete (RTC) Network Flow
Processor (NFP)~\cite{nfp} and use Micro-C~\cite{micro-c} (preferred) or P4 for
programming. NFP contains nearly a hundred cores and several hardware
accelerators. The cores are arranged in several islands and equipped with a
hierarchical memory structure.

\noindent{\textbf{Compatibility.}}
For Netronome NFP, it cannot support floating-point operation while multiplication and division of integers are allowed:
\begin{equation}
    \bigvee_{v \in \mathcal{B}_{CA}+\mathcal{APF}; d \in NFP} x_{v,d} = False
\end{equation}

\noindent{\textbf{FPC allocation.}}
FPCs support parallel as well as pipeline working mode. For any instruction of the program to be allocated, it can be placed on some FPCs and these FPCs will work in parallel. Then any of the other unplaced instructions will either be put on all of these FPCs or none of them, i.e., any two instructions can only co-exist on the same group of FPCs in parallel or completely independent of each other. 
Thus, we have the constraint as:
\begin{equation}
    \bigwedge_{p_k, p_l\in V} [ (\sum_{i\in L_d}\sum_{c\in C_i} a_{p_k,i,c}a_{p_l,i,c} = 0)\vee (\bigwedge_{i\in L_d, c\in C_i} a_{p_k,i,c} = a_{p_l,i,c}) ]
\end{equation}
where $L_d$ is the set of FPC islands, $C_i$ is the set of FPCs on the island $i$, $V$ is the whole set of IR instructions, $a_{p,i,c}\in \{0, 1\}$ (0: no, 1: yes) indicates whether instruction $p$ is placed on the FPC $c$ of the island $i$. 

\noindent{\textbf{FPC capacity.}}
Every FPC in a NFP are only capable to support a fixed number of Micro-Instructions (MIs), and thus we have the constraint as:
\begin{equation}
    \bigwedge_{i\in L_d, c \in C_i} \big [ \sum_{v\in V} \sum_{p\in v} x_{v,d} a_{p,i,c}n_{ins}(p) \leq N_{ins} \big ]
\end{equation}
where $n_{ins}(p)$ denotes the number of MIs mapped to $p$, and $N_{ins}$ is the total number of MIs supported by the FPC.

\noindent{\textbf{Memory.}}
NFP adopts a hierarchical memory structure: i.e., the general purpose register (GPR, owned by each thread in a FPC and cannot be shared), the local memory (LM, owned by each FPC but can be locally shared among threads on a FPC), the cluster local scratch (CLS, can be shared among different FPCs on the same island), the cluster target memory (CTM, has the same sharing property with CLS but has a larger size and higher latency), the internal memory (IM, can be shared globally among all FPCs), external memory(EM, the off-chip memory that can be globally shared and has larger size but higher latency than IMEM). Particularly, for a defined data structure in Micro-C program, we can manually specify which memory the data is to be placed in and also have options to decide whether the memory is shared. If the memory is chosen not to be shared, the data will be copied for all threads which redundantly consume memory but enjoys lower processing latency. To describe all kinds of data placement decisions, we use variable $m_{p, r}$, $m_{p, lm}$, $m_{p, cls}$, $m_{p, ctm}$, $m_{p, im}$, $m_{p, em} \in \{0, 1\}$ (0: no, 1: yes) to indicate whether the data of instruction $p$ should be located in the memory of GPR, LM, CLS, CTM, IMEM and EMEM respectively. Besides, the data can be shared only in one of the three methods: 1) shared by threads in each local FPC; 2) shared among different FPCs in a local island; 3) globally shared across all islands. We use $f_{p, st}$, $f_{p, si}$, $f_{p, sg} \in \{0, 1\}$ (0: no, 1: yes) to represent whether the data is shared by the above three methods, and we have
\begin{equation}
    f_{p, st} + f_{p, si} + f_{p, sg} \leq 1
\end{equation}
where the sum of $f$ equaling to 0 indicates the data is not shared.

For all instructions to be allocated on the NFP, the related data (e.g., internal variable) is only allowed to be placed on one kind of the above several memories. To efficiently support the large match-action table, we can manually divide a whole table into two parts and allocate them across memories, i.e., a small part of popular entries on fast but low-capacity memory and the remained part on slow but large memory.
Thus, we have the constraints:
\begin{equation}
    \begin{aligned}
        \bigwedge_{v\in V-\mathcal{B}_{NEM}-\mathcal{B}_{SNEM}; p\in v} \big [x_{v,d}\varphi_d(p) \leq \sum_{k\in K_m} m_{p,k}\big ] = True
    \end{aligned}
\end{equation}
where $\varphi_d(p)\in \{0,1\}$ (0: no, 1: yes) denotes whether instruction $p$ creates a new data, and $K_m$ is all kinds of memories can be used. Correspondingly, each memory has its capacity constraint. For GPR, it is owned by each thread of each FPC, and cannot be shared. Besides, due to the very limited size of GPR, complex match-action tables are also not allowed. Thus we have:
\begin{equation}
\bigwedge_{i\in L_d;c\in C_i} \big [\sum_{v\in V-B_M}\sum_{p\in v} \sum_{e \in \mathcal{D}_s(p)}m_{p,r}h(e)\lceil \frac{w(e)}{w_{R}} \rceil \leq \frac{N_{R}}{n_{t}} \big]
\end{equation}
where $B_M=\mathcal{B}_{EM}+\mathcal{B}_{NEM}+\mathcal{B}_{SEM}+\mathcal{B}_{SNEM}$, $n_{t}$ denotes the number of threads working on the FPC, $\mathcal{D}_s(p)$ is the set of data to be stored for instruction $p$, $h(e)$ and $w(e)$ are depth and width of data $e$, $w_{R}$ is the width of GPR memory unit, and $N_{R}$ (only a few hundreds) is the total number of GPR units in each FPC . For LM, it has multi-kilobyte size and can be locally shared among threads while non-exact-match is still not allowed. Thus we have:
\begin{equation}
    \bigwedge_{i\in L_d; c\in C_i} \big [ \sum_{v\in V-B_{NE}}\sum_{p\in v} \sum_{e \in \mathcal{D}_s(p)} \frac{m_{p,lm} n_{t}}{1+f_{p,st}(n_{t}-1)} h(e) \lceil \frac{w(e)}{w_{LM}} \rceil \leq N_{LM} \big]
\end{equation}
where $B_{NE}=\mathcal{B}_{NEM}+\mathcal{SNEM}$, $w_{LM}$, $N_{LM}$ are the width and the number of LM units in each FPC, $\frac{m_{p,lm}\cdot n_{t}}{1+f_{p,st}(n_{t}-1)}$ equals to 1 when data of IR instruction $p$ is shared by all threads and $n_{t}$ otherwise. For CLS (tens of kilobytes) and CTM (hundreds of kilobytes), they both can be shared within island (among all threads on all FPCs), and thus we have:
\begin{equation}
    \begin{aligned}
         &\bigwedge_{k\in \{CLS, CTM\}; i\in L_d} \big [ \sum_{v\in V}\sum_{p\in v} \sum_{e \in \mathcal{D}_s(p)} \Gamma_m h(e) \lceil \frac{w(e)}{w_{k}} \rceil 
         \leq N_{k} \big] \\ 
         & \Gamma_m = \frac{m_{p,k}g_{p,i} n_tn_c}{1+f_{p,st}(n_{t}-1)+f_{p,si}(n_cn_t-1)+f_{p,sg}(n_t n_c-1)}
    \end{aligned}
\end{equation}
where $n_c = \sum_{c\in C_i}x_{v,d}a_{p,i,c}$ is the number of FPCs where the instruction $p$ is located, and $g_{p,i}\in \{0,1\}$ is dedicated for globally sharing mode of memory and represents whether the data for $p$ is stored on the island $i$. For IM and EM, they are outside of island and can be shared by all FPCs on chip. Thus, we have the constraint:
\begin{equation}
    \begin{aligned}
    & \bigwedge_{k\in \{IM, EM\}} \big[\sum_{v\in V}\sum_{p\in v} \sum_{e \in \mathcal{D}_s(p)} \Gamma_m' h(e) \lceil \frac{w(e)}{w_{k}} \rceil \leq N_{k}\big] \\ 
    & \Gamma_m' = \frac{m_{p,k} g_{p,i} n_t \sum_{i\in L_d} \frac{n_c}{1+f_{p,si}(n_c -1)}}{1+f_{p,st}(n_{t}-1)+f_{p,si}(n_t-1)+f_{p,sg}(n_{t}\sum_{i\in L_d}n_c-1)}
    \end{aligned}
\end{equation}
where IM is a still on-chip SRAM with multiple megabytes, and EM is an off-chip DRAM with multiple Gigabytes.

Further, for non-exact-match table, it is not allowed on GPR and LM, and also has to follow the constraint on width of the matching key:
\begin{equation}
    \bigwedge_{v\in \mathcal{B}_{NEM}; p\in v; k \in \{CLS, CTM, IM, EM\}} \big (x_{v,d} m_{p,k}w_{key}(p) \leq W_k \big)
\end{equation}
where $w_{key}(p)$ is the key width of non-exact-match table, $W_k$ is the width limit which is 32b for CLS and 512b for the other memories.

\noindent{\textbf{Accelerator.}}
NFP relies on pre-implemented dedicated accelerators for some complex operations, e.g., hash, and only support a limited number of operations.
\begin{equation}
    \bigwedge_{k_p\in \mathbf{U}(\mathcal{B}_{HF}+\mathcal{B}_{TF})} \big[ \sum_{v\in \mathcal{B}_{AF}+\mathcal{B}_{CF}}\sum_{p \in v} x_{v,d}\wedge(p=k_p) \leq N_{acc}(k_p)\big]
\end{equation}
where $\mathbf{U}(\mathcal{B}_{HF}+\mathcal{B}_{TF})$ is the set of unique operation types for instruction blocks, and $N_{acc}(k_p)$ is the number of provided accelerators for each operation.

\subsection{Xilinx FPGA}
\label{appendix:cons_fpga}
Xilinx FPGA is used in smartNICs
(e.g., SN1000 \cite{SN1000} and acceleration cards (e.g., Xilinx Alveo
U280~\cite{U280}). SmartNIC supports P4 for pipeline programming via
VNetP4~\cite{VNetP4} toolchain (formerly known as SDNet).  By using
\texttt{UserExtern}, VNetP4 allows user to integrate extra functions into the
main P4 program.  Acceleration card uses Verilog HDL or HLS, a C-like
language)~\cite{hls} for programming.  The resources in Xilinx FPGA include
LUTs, DSP slices, Flip-flops (FF), UltraRAMs (URAM), Block RAMs (BRAM), and
possible High-Bandwidth Memory (HBM).
In our work, we follow the principle that uses RAM in preference to LUT for matching tables and keeps all the other default configurations. 

\noindent{\textbf{Memory.}}
For matching tables, the width has a predefined limitation and the depth depends on the available resources. Thus, the width constraint is modeled as:
\begin{equation}
    \bigwedge_{v \in V_m; p \in v} \big[x_{v,d} a_{p,s} \cdot I_{(w_{key}(p) \leq W_k)} \cdot I_{(w_{val}(p) \leq W_v)}>0 \big]
\end{equation}
where $I_{(x)}$ is the indicator function which equals to 1 if condition $x$ is satisfied and 0 otherwise, $V_{m}=\mathcal{B}_{EM}+\mathcal{B}_{NEM}+\mathcal{B}_{DM}+\mathcal{B}_{SEM}+\mathcal{B}_{SNEM}$ denotes all kinds of matching tables, $w_{key}(p)$ and $w_{val}(p)$ are the width of table key and value, respectively, and $W_k$ and $W_v$ are their limitations, respectively.

Time-Division Multiplexing (TDM) factor $f_{TDM}\in \{1,2,4\}$ plays a key role in BRAM/URAM allocation for implementing matching tables. $f_{TDM}$ is proportional directly to RAM frequency but inversely to table lookup rate. A search operation on CAM requires at least four RAM accesses, which can be performed in parallel using four RAM instances with $f_{TDM}=1$ or sequentially using four accesses in the same RAM with $f_{TDM}=4$. BRAM is preferred to be used to instantiate relatively small tables, and we have:
\begin{equation}
    \sum_{v \in V_{m}} \sum_{p \in v; h(p) \leq \theta_{bm}} x_{v,d}a_{p,s} \frac{4}{f_{TDM}} \lceil \frac{w(p)}{w_{bm}} \rceil \leq N_{bram}
\end{equation}
where $N_{bram}$ is the number of BRAM blocks, $\theta_{bm}=\frac{4\alpha h_{bm}}{f_{TDM}}$ limits the value of $h(p)$ (i.e., the depth of matching table that can be implemented by BRAM), $\alpha$ is the depth utilization ratio of BRAM/URAM (e.g., 95\% for BCAM) due to the storage efficiency, $h_{bm}$ and $w_{bm}$ are the depth and width of a BRAM block, respectively, and $w(p)=w_{key}+w_{val}(p)+w_{extra}$ is the practical width of table entry plus some extra bits such as the valid flag.

URAM is typically used for large tables. If $h(p)$ exceeds $\theta_{bm}$, URAM will always be used unless the table is so large that only HBM can hold. So we have the constraint:
\begin{equation}
    \sum_{v \in V_m} \sum_{p \in v; h(p)>\theta_{bm}} x_{v,d}a_{p,s} \frac{4}{f_{TDM}} \lceil \frac{h(p)}{\frac{4}{f_{TDM}} \alpha h_{um}} \rceil \lceil \frac{w(p)}{w_{um}} \rceil \leq N_{uram}
\end{equation}
where $h_{bm}$ and $w_{bm}$ are depth and width of a URAM block, respectively. Here each table can use at most 320 URAMs if the RAM frequency is lower than 400MHz and the number is reduced to half if the RAM frequency is over 400MHz or HBM is used. Since HBM is as large as 8GB, we omit its constraint. 

\noindent{\textbf{DSP.}}
For efficient resource usage, simple logic (e.g., integer add, subtraction, bit operation, and RAM managing logic) is implemented using LUTs, while complex logic (e.g, multiply and floating arithmetic) uses DSPs. Thus, we have the following constraints on DSP slices:
\begin{equation}
    \sum_{v \in B_C} \sum_{p \in v}x_{v,d}a_{p,s} d(p) \leq N_{dsp}
\end{equation}
where $B_c = \mathcal{B}_{IC}+\mathcal{B}_{CA}$, $d(p)$ represents the needed number of DSPs for implementing primitive $p$, and $N_{dsp}$ denotes the total number of DSP slices on an FPGA.

\noindent{\textbf{LUT.}}
LUT, in addition to implementing logic circuits, can be used as Distributed RAM in a unit form called SLICEM to implement small DCAM in VNetP4. A SLICEM contains 8 LUTs with each being configured as 32x2 or 64x1. Thus a SLICEM can realize single-port RAM configurable from 32x(1 to 16)-bit to 512x1-bit, dual-port RAM configurable from 32x(1 to 8)-bit to 256x1-bit, and quad-port RAM from 32x(1 to 4)-bit to 128x1-bit. As the basic resource in FGPA, LUT usage is hard to estimate. Therefore, to ensure feasibility, we formulate the LUT constraint by a maximum utilization ratio $\beta$ (about 75\% empirically) to all LUTs $N_{lut}$:
\begin{equation}
\label{eq:lut}
    \begin{aligned}
          &\sum_{v \in \mathcal{B}_{DM}}\sum_{p \in v; h(p)\leq \theta_{dr}} 8x_{v,d}a_{p,s} \big[ I_{(w(p)-w_l\geq 0)} \lceil \frac{h(p)}{h_l}\rceil \lceil \frac{w(p)}{w_l} \rceil + \\ & I_{(w(p)-w_l<0)}\lceil\frac{h(p)}{h_l\lfloor \frac{w_l}{w(p)}\rfloor}\rceil \big]  + \sum_{v \in V-\mathcal{B}_{DM}} \sum_{p \in v}x_{v,d}a_{p,s} l(p) \leq \beta N_{lut}
    \end{aligned}
\end{equation}

The first part of Eq.~\ref{eq:lut} is the total usage of LUTs for Distributed RAM and the second part is for simple logic. $h(p)$ is the number of entries of table in DCAM and $\theta_{dr}$ specifies the ceiling value of $h(p)$ for using Distributed RAM as DCAM. $w(p)$ denotes the entry width, $w_l$ is the maximum configurable width of Distributed RAM, $h_l$ corresponds to the depth (e.g., $w_l=16$ and $h_l=32$ for single-port RAM), $\delta(\cdot)$ is the Unit Step Function (i.e., the number of needed LUTs is $\lceil \frac{h(p)}{h_l}\rceil \lceil \frac{w(p)}{w_l} \rceil$ when $w(p)>w_l$, or $\lceil\frac{h(p)}{h_l\lfloor \frac{w_l}{w(p)}\rfloor}\rceil$ otherwise), and $l(p)$ is the number of needed LUTs for implementing primitive $p$.

\noindent{\textbf{Flip-flop.}}
Flip-flops (FFs) are widely used for temporary storage between different timing logic or for tiny-TCAM/tiny-BCAM tables. We have the following constraint:
\begin{equation}
    \sum_{v \in B_{tiny}} \sum_{p \in v} x_{v,d}a_{p,s}h(p)w(p) + \sum_{v \in V} \sum_{f\in v} x_{v,d}a_{p,s}w(f) \leq N_{reg}
\end{equation}
where $B_{tiny}=\{V_m|_{h(p)\leq 32}\}$, $h(p)$ and $w(p)$ are the depth and width of matching table $p$, respectively, $f \in v$ represents the internal variables of width $w(f)$ appearing in block program $v$, and $N_{reg}$ denotes the number of all available FFs.

\section{Deployment Constraints}
\label{appendix:deploy_constraint}
Each block $v$ should be allocated only once, and each instruction in the block should be deployed:

\begin{equation}
    \bigwedge_{v \in V} \big[\sum_{d \in D} x_{v,d}\bigwedge_{p\in v}(\bigvee_{s\in S_d}a_{p,s})= 1 \big]
\end{equation}
where $x_{v,d}$ indicates whether primitive block $v$ is deployed on device $d$, and $a_{p,s}$ denotes whether primitive $p$ is deployed on stage $s$.

Since the application throughput is bottlenecked at the device with the minimal bandwidth, given the throughput requirement $H$, we have the constraint:
\begin{equation}
    \bigvee_{l\in L}\big[\bigwedge_{d\in D[l]} (h(d) \geq H[l])\big]
\end{equation}
where $h(d)$ is the bandwidth of device $d$.

Typically the application flow has a fixed forwarding path, which raises two topology constraints: deployment scope $T_s$ (i.e., the blocks can only be allocated on devices along the path), and deployment direction (i.e., the block execution sequence should match the packet forwarding direction). The scope constraint is: 
\begin{equation}
    \sum_{v\in , d\notin T_s} \ x_{v,d}=0
\end{equation}
and the direction constraint is:
\begin{equation}
\label{eq:block_dp}
    \bigwedge_{d_i,d_j \in D; v_k, v_l \in V} (F_{d_i, d_j}R_{v_k, v_l}x_{v_k, d_i}x_{v_l, d_j} \geq 0)   
\end{equation}

In the equation, $F_{d_i, d_j}$ denotes the deployment direction: 1 represents the forwarding direction, -1 vice versa, and 0 means no direction needs to be enforced (e.g., for an FPGA-based acceleration card attached to a switch); $R_{v_k, v_l}$ denotes the dependency between two blocks: 1 represents that $v_l$ relies on $v_k$, -1 vice versa, and 0 means $v_k$ and $v_l$ are independent.

Similarly, dependent blocks on the same device should conform to the pipeline direction:
\begin{equation}
\label{eq:primitive_dependency}
        \bigwedge_{s\in S_d; p_i, p_j \in v} (R_{p_i, p_j}a_{p_i,s}a_{p_j,s} = 0)
\end{equation}
The constraint ensures that no stage overlap occurs in the case that $v_j$ depends on $v_i$.

Furthermore, the dependent primitives cannot be placed on the same pipeline stage. For Tofino series chips, there is a particular circumstance: a non-matching-table primitive can be placed in the same stage with the matching table that it depends on, to construct a match-action structure as long as they share the same conditional statement. We use $R_{p_i,p_j}$ to denote the dependency between $p_i$ and $p_j$, where 1 represents $p_j$ depends on $p_i$, -1 vice versa, and 0 means $p_i$ and $p_j$ are independent. The primitive dependency constraint is:
\begin{equation}
\label{eq:match_action}
    \begin{aligned}
        &\bigwedge_{s \in S_d; p_i,p_j \in v} \big[(R_{p_i, p_j}a_{p_i,s}a_{p_j,s}=0 )\big] \\
        &\bigwedge_{d\in Tofino} R_{p_i, p_j}a_{p_i,s}a_{p_j,s}(1-case) = 0
    \end{aligned}
\end{equation}
where $case$ is $(R_{p_i, p_j}=1)\wedge(p_i \in \mathcal{B}_{EM}+\mathcal{B}_{NEM})\wedge (p_j \notin \mathcal{B}_{EM}+\mathcal{B}_{NEM})\wedge (cond(p_i)=cond(p_j)))$, and $cond(\cdot)$ denotes the requirement on conditional statement. Through the above mapping, primitives satisfied with the constraint can be constructed as match-action table structure.

Pipeline switch have separate pipelines $\xi_{ig}$ and $\xi_{eg}$ for ingress and egress, respectively (i.e., $S_d = \xi_{ig}+\xi_{eg}$). The instructions related to forwarding decision can only be deployed at ingress. Denoting these instructions as $\mathcal{F}_{fd}$, we have:
\begin{equation}
    \sum_{p\in \mathcal{F}_{fd}, s\in \xi_{eg}} a_{p,s} = 0
\end{equation}

\end{document}